\documentclass[journal]{IEEEtran}
\ifCLASSINFOpdf
\else
\fi
  \usepackage[caption=false,font=footnotesize]{subfig}
\usepackage{graphicx}
\usepackage{amsmath,amssymb,amsfonts,amsthm}
\usepackage[utf8]{inputenc}
\usepackage{multirow}
\usepackage{booktabs}
\usepackage{tabularx}
\usepackage{graphicx}
\usepackage{enumitem}
\usepackage[keeplastbox]{flushend}
\usepackage{algorithm}
\usepackage{algorithmic}

\newtheorem{theorem}{Theorem}
\newtheorem{lemma}{Lemma}
\newtheorem{corollary}{Corollary}
\newtheorem{definition}{Definition}

\usepackage[svgnames]{xcolor}
\newcommand{\new}[1]{\textcolor{black}{#1}}
\newcommand{\moved}[1]{\textcolor{black}{#1}}

\newcommand{\newnew}[1]{\textcolor{black}{#1}}

\begin{document}
%
\title{Channel estimation: unified view of optimal performance and pilot sequences}
%
%
%

\author{Luc Le Magoarou, St\'ephane Paquelet
\thanks{Luc Le Magoarou and St\'ephane Paquelet are both with bcom, Rennes, France. Contact addresses:  \texttt{luc.lemagoarou@b-com.com}, \texttt{stephane.paquelet@b-com.com}.}}

%
%

\markboth{Accepted version}%
{Accepted version}
%



\maketitle

\begin{abstract}
Channel estimation is of paramount importance in most communication systems in order to optimize the data rate/energy consumption tradeoff.
In modern systems, the possibly large number of transmit/receive antennas and subcarriers makes this task difficult.
\new{Designing pilot sequences of reasonable size yielding good performance is thus critical.
Classically, the number of pilots is reduced by viewing the channel as a random vector and assuming knowledge of its distribution. In practice, this requires estimating the channel covariance matrix, which can be computationally costly and not adapted to scenarios with high mobility.
In this paper, an alternative view is considered, in which the channel is a function of unknown deterministic parameters. In this setting, the problem of designing optimal pilot sequences of smallest possible size is studied for \emph{any} parametric channel model. To do so, the Cram\'er-Rao bound (CRB) for this general channel estimation problem is given, highlighting its key dependency on the introduced \emph{variation space}. Then, the minimal size of pilot sequences and minimal value of the CRB are determined. Moreover, a general strategy to build optimal minimal length power constrained pilots sequences is given, based on an estimation of the variation space. 
The theoretical results are finally illustrated in a massive MIMO system context. They conveniently allow to retrieve well known previous results, but also to exhibit minimal length optimal pilot sequences for a new strategy based on a nonlinear physical model.
}
\end{abstract}


%
\IEEEpeerreviewmaketitle

\section{Introduction}
%
%
%
%

\IEEEPARstart{C}{ommunication} systems make use of a physical channel to convey information between a transmitter and a receiver \cite{Haykin2008}. Knowing the channel state at both ends of the link allows to maximize the data rate, hence the need to estimate the channel. This can be carried out by sending pilot signals known by both the transmitter and the receiver to gather noisy observations used to estimate the channel. 

Recently, the ever-growing need for data rate in modern communication networks led to use channels of very high dimension, which makes channel estimation difficult. For example, it has been recently proposed to use massive multiple input multiple output (massive MIMO) wireless systems \cite{Rusek2013,Larsson2014,Lu2014} with a large number of transmit and receive antennas in the millimeter-wave band \cite{Rappaport2013, Swindlehurst2014}, where a large bandwidth can be exploited. In that case the channel comprises hundreds or even thousands of complex numbers, whose estimation is a very challenging signal processing problem \cite{Heath2016}.

\new{
Designing pilot sequences that lead to low estimation error and are of reduced size (compared to the channel dimension) is thus a critical issue in massive MIMO systems. Classically, it has been done by considering the channel as a random vector whose distribution is known a priori, which naturally leads to the use of bayesian methods and estimators such as the linear minimum mean squared error (LMMSE). The minimal size of the pilot sequence is then determined by the effective rank of the channel covariance matrix \cite{Biguesh2006}, and efficient strategies such as the joint spatial division and multiplexing (JSDM) \cite{Adhikary2013, Adhikary2014} can be implemented. However, the main drawback of such methods is that it requires to estimate the channel covariance matrix, which can be computationally costly and unfit for high mobility scenarios (since the channel covariance then changes fast).}

\new{
Another solution to envision channel estimation, which does not require covariance estimation, is to consider the channel as a function of parameters being deterministic unknown quantities, such as the channel coefficients or the directions and complex gains of the most significant propagation paths. This naturally leads to classical estimators based on the maximum likelihood (ML) principle. Following this line of thought, modern approaches have emerged \cite{Gao2015} \cite{Xie2017} that exploit some prior knowledge regarding the parameters to estimate in order to design efficient transmission strategies. In this setting,
which quantity does determine the minimal size of pilot sequences? What is the best attainable performance? How to design optimal pilot sequences? Based on which a priori information?
}
\new{
\newline{\bf Contributions.} In this paper, we tackle these questions in a general unified way, for \emph{any} parametric channel model (linear or not). Based on the 
 Cram\'er-Rao bound (CRB) \cite{Rao1945,Cramer1946} of the considered problem, we show that the crucial object for pilot sequences determination is the \emph{variation space} of the channel, which is a notion we introduce. Identifiability conditions, the minimal size of pilot sequence, the minimal attainable variance and a strategy to build optimal pilot sequences of minimal size are given, all based on the variation space. We argue that the variation space is an object whose estimation may be simpler than that of the covariance. The theoretical part of the paper (which constitute the main contribution) is then illustrated on several MIMO channel models, and it allows to determine optimal pilot sequences and optimal performance for a new promising channel estimation strategy we propose for MIMO systems that operate in frequency division duplex (FDD) mode.} 

\noindent {\bf Related work.}
\new{On the theoretical side, this paper is a generalization and an unification of} many results obtained in the case of linear deterministic channel models \new{(in which the model parameters are simply the channel coefficients)}, both in a MIMO context \cite{Biguesh2006,Marzetta1999} and for multicarrier systems \cite{Ma2003,Barhumi2003,Minn2006}. \new{ Indeed, the present analysis based on the variation space allows to treat simple linear models and more elaborate nonlinear physical channel models the same way.} Another significant difference with prior work is that the analysis of the present paper is based on the CRB (which depends only on the model and not on the estimation method) and not directly on the error incurred by a specific estimator.

There is also a vast body of literature regarding optimal pilot sequences in a bayesian channel estimation setting, for which the channel is assumed to follow a known Gaussian \cite{Kotecha2004b,Bjornson2009,Choi2014}, \new{or more elaborate Gaussian mixture \cite{Gu2019} distribution.} 
These approaches are different in nature from the one of this paper, since  (i) they consider a specific estimator (the linear minimum mean squared error (LMMSE)), and (ii) their objective is to minimize the estimation error in average over the channel estimated distribution. On the other hand, the analysis of the  present paper is estimator independent and its objective can be seen as the minimization of the error for a given channel realization.
Recently, it has also been proposed to look for optimal pilot sequences in a multi-user bayesian setting. In \cite{Bazzi2017}, sequences are found by numerical optimization, minimizing a weighted sum of the channel estimation errors of each user. In \cite{Bazzi2018} and \cite{Bazzi2019}, heuristics are proposed which amount to send pilot sequences that span the union of the spaces generated by the leading eigenvectors of the channel correlation matrices of all users. 

\new{On the practical side, the analysis performed in this paper allows us to suggest a new transmission strategy for massive MIMO systems operating in FDD mode. It relies on the physical assumptions that the angles of arrival for channel propagation paths vary slowly and are reciprocal between the uplink and the downlink. This assumption is also at the origin of recent proposals \cite{Gao2015} \cite{Xie2017}. The strategy we propose is similar to this prior work in that is uses previous angle estimates (indifferently acquired in the uplink or downlink) to design pilot sequences. However, it is different since it allows to reestimate the angles at each step (with a small additional overhead), which leads to better performance lower bounds, as shown in section~\ref{sec:illustration}.}


\noindent {\bf Organization of the paper.} The studied problem is formulated in section~\ref{sec:formulation}. \new{The notion of variation space is introduced, and an expression of the  Cram\'er-Rao bound (CRB) based on it is given in section~\ref{sec:CRB_variation_space}.} Identifiability conditions on the observation matrices and the minimal number of observations for which they can be fulfilled are given in section~\ref{ssec:identifiability}. In section~\ref{ssec:optimality}, we express the minimal variance of any unbiased estimator by optimizing the CRB under a power constraint on the observation matrix. Associated observation matrices of minimal length are also exhibited \new{in section \ref{ssec:design}, as well as an algorithm to build it based on an estimation of the variation space.}
These results are illustrated in section~\ref{sec:illustration}, where it is shown that the proposed theoretical framework allows to retrieve well-known results previously established for linear models, \new{but also to propose a new efficient transmission strategy for massive MIMO systems operating in FDD mode. }
 \new{ For convenience and in order to keep the flow of the paper, most technical proofs are given in appendix.}

Note that this paper is partially based on some of our previous work \cite{Lemagoarou2018,Lemagoarou2019}, in which the Cram\'er-Rao bound in the specific case of a physical channel model was stated. The novelty of this paper is that the Cram\'er-Rao bound is here optimized, and the derivation is more general since it is valid for any parametric model.

\section{Problem formulation}
\label{sec:formulation}

\noindent{\bf Notations.} Matrices and vectors are denoted by bold upper-case and lower-case letters: $\mathbf{A}$ and $\mathbf{a}$ (except 3D ``\emph{spatial}'' vectors that are denoted $\overrightarrow{a}$); the $i$th column of a matrix $\mathbf{A}$ by $\mathbf{a}_i$; its entry at the $i$th line and $j$th column by $a_{ij}$. $\mathbf{A}_{[i:j,:]}$ denotes the matrix built taking the rows $i$ to $j$ of $\mathbf{A}$ (matlab style indexing). 
A matrix transpose, conjugate and transconjugate is denoted by $\mathbf{A}^T$, $\mathbf{A}^*$ and $\mathbf{A}^H$ respectively. The trace of a linear transformation represented by $\mathbf{A}$ is denoted $\text{Tr}(\mathbf{A})$.
 The linear span of a set of vectors $\mathcal{A}$ and its dimension (if it is a vector space) are denoted: $\text{span}_\mathbb{R}(\mathcal{A})$ and $\text{dim}_\mathbb{R}(\mathcal{A})$ when considering linear combinations with real coefficients, or $\text{span}_\mathbb{C}(\mathcal{A})$ and $\text{dim}_\mathbb{C}(\mathcal{A})$ when considering linear combinations with complex coefficients. The orthogonal complement of a subspace $\mathcal{W}$ is denoted $\mathcal{W}^\perp$. The Kronecker product is denoted by $\otimes$.
 The identity matrix is denoted by $\mathbf{Id}$. $\mathcal{CN}(\boldsymbol\mu,\boldsymbol{\Sigma})$ denotes the standard complex gaussian distribution with mean $\boldsymbol\mu$ and covariance $\boldsymbol{\Sigma}$. $\mathbb{E}(\cdot)$ denotes the expectation and $\text{cov}(\cdot)$ the covariance of its argument. 
\vspace{5mm}

\subsection{Observations}
\label{ssec:observations}
We consider the general channel estimation setting where a channel $\mathbf{h} \in \mathbb{C}^{N_d}$ is to be estimated, $N_d$ being the total number of complex dimensions of the channel. \new{For example, in the case of a channel between $N_t$ transmit antennas and $N_r$ receive antennas on $N_f$ subcarriers, we have $N_d = N_rN_tN_f$.} We assume it is deterministic and follows a parametric model depending on $N_p$ real parameters. It can then be seen as a function $\mathbf{h}:\mathbb{R}^{N_p} \rightarrow \mathbb{C}^{N_d}$ that maps each parameters value to the corresponding channel (we will denote the channel indifferently $\mathbf{h}$ or $\mathbf{h}(\boldsymbol{\theta})$ depending on the context). Note that this is the most generic setting since complex parameters can always be decomposed into real and imaginary parts (or modulus and angle) and thus correspond to two real parameters each. The only assumption that we make about the channel model is that the function $\mathbf{h}$ is differentiable with respect to the parameters. \moved{We denote
$
\frac{\partial \mathbf{h}}{\partial \boldsymbol{\theta}}
\triangleq
\big(\frac{\partial \mathbf{h}}{\partial \theta_1},\dots,\frac{\partial \mathbf{h}}{\partial \theta_{N_{p}}} \big) \in \mathbb{C}^{N_d\times N_{p}}
$
the complex gradient of the channel with respect to its real parameters.}

Estimation is made based on $N_m$ noisy linear observations of the form
\begin{equation}
\mathbf{y} = \mathbf{M}^H\mathbf{h} + \mathbf{n},
\label{eq:observation}
\end{equation}
where $\mathbf{n} \in \mathbb{C}^{N_m}$ corresponds to the noise whose entries are assumed i.i.d.\ complex gaussians of variance $\sigma^2$, so that $\mathbf{n}\sim \mathcal{CN}\big(\mathbf{0},\sigma^2\mathbf{Id}_{N_m} \big)$, and $\mathbf{M}\in \mathbb{C}^{N_d\times N_m}$ is \new{the matrix representing the measurement process, that we hereafter denote the \emph{observation matrix}.} It is entirely determined by the pilot sequences sent by the transmitter and the combining operations done at the receiver.
this way of expressing the observations is very general (an example in a generic MIMO wideband context is given in section~\ref{sec:illustration}).

\subsection{Estimation}
\label{ssec:esterror}
The estimator of the parameters is a function mapping obervations to estimates, denoted $\hat{\boldsymbol{\theta}} :\mathbb{C}^{N_m} \rightarrow \mathbb{R}^{N_p}$. The estimate $\hat{\boldsymbol{\theta}}(\mathbf{y})$ will be denoted $\hat{\boldsymbol{\theta}}$ (as the estimator) for shorter notations. The channel estimate is given by the model function, as $\mathbf{h}(\hat{\boldsymbol{\theta}})$.
The error is measured by the mean squared error (MSE):
\begin{align*}
\text{MSE}(\hat{\boldsymbol{\theta}}) &\triangleq \mathbb{E}\left[ \big\Vert \mathbf{h}(\boldsymbol{\theta}) - \mathbf{h}(\hat{\boldsymbol{\theta}}) \big\Vert_2^2 \right] \\
&=\big\Vert\mathbf{h}(\boldsymbol{\theta})-\mathbb{E}[\mathbf{h}(\hat{\boldsymbol{\theta}})]\big\Vert_2^2
+
\mathbb{E}\left[\big\Vert \mathbf{h}(\hat{\boldsymbol{\theta}})-\mathbb{E}[\mathbf{h}(\hat{\boldsymbol{\theta}})]\big\Vert_2^2\right],
\label{eq:MSE}
\end{align*}
where the expectation is taken over the noise distribution, and the second line corresponds to the well-known bias-variance decomposition  \cite{Kay1993}.
We assume throughout the paper that the considered channel estimators are unbiased with respect to $\mathbf{h}(\boldsymbol{\theta})$, which reads $\mathbb{E}[\mathbf{h}(\hat{\boldsymbol{\theta}})] = \mathbf{h}(\boldsymbol{\theta})$. 
It follows
\begin{equation}
\text{MSE}(\hat{\boldsymbol{\theta}}) = 
\mathbb{E}\left[\big\Vert \mathbf{h}(\hat{\boldsymbol{\theta}})-\mathbb{E}[\mathbf{h}(\hat{\boldsymbol{\theta}})]\big\Vert_2^2\right]
=
\text{Tr}[\text{cov}(\mathbf{h}(\hat{\boldsymbol{\theta}}))].
\label{eq:MSE_unbiased}
\end{equation}
This way, the bias is null and the MSE is entirely due to the variance of the estimator $\mathbf{h}(\hat{\boldsymbol{\theta}})$.

\moved{
\subsection{Cram\'er-Rao bound}
\label{ssec:CRB1}
The variance of any unbiased estimator is bounded below by the Cram\'er-Rao bound \cite{Rao1945,Cramer1946}, so that $$
\text{MSE}(\hat{\boldsymbol{\theta}}) = \text{Tr}\Big[\text{cov}\big(\mathbf{h}(\hat{\boldsymbol{\theta}})\big)\Big]
\geq
\text{CRB}(\boldsymbol{\theta},\mathbf{M}),
$$ 
where the complex CRB \cite{Vandenbos1994} takes the form
\begin{equation}
\text{CRB}(\boldsymbol{\theta},\mathbf{M})
\triangleq
\text{Tr}\bigg[\frac{\partial \mathbf{h}}{\partial \boldsymbol{\theta}} \mathbf{I}(\boldsymbol{\theta},\mathbf{M})^{-1}\frac{\partial \mathbf{h}}{\partial \boldsymbol{\theta}}^H\bigg],
\label{eq:CRB_start}
\end{equation}
$\mathbf{I}(\boldsymbol{\theta},\mathbf{M}) \in \mathbb{R}^{N_p\times N_p}$ being the Fisher information matrix (FIM) which quantifies the amount of information about the parameters $\boldsymbol{\theta}$ that the observation $\mathbf{y}$ carries when using the observation matrix $\mathbf{M}$.
The observation defined in \eqref{eq:observation} follows a gaussian distribution,
$$
\mathbf{y} \sim \mathcal{CN}\big(\mathbf{M}^H \mathbf{h},\sigma^2\mathbf{Id} \big),
$$
so that the FIM is given by the Slepian-Bangs formula \cite{Slepian1954,Bangs1971, Besson2013}: 
\begin{equation}
\mathbf{I}(\boldsymbol{\theta},\mathbf{M})= \frac{2}{\sigma^2}\mathfrak{Re}\bigg\{\frac{\partial \mathbf{h}}{\partial \boldsymbol{\theta}}^H \mathbf{M}\mathbf{M}^H \frac{\partial \mathbf{h}}{\partial \boldsymbol{\theta}}\bigg\}.
\label{eq:slepian}
\end{equation}
This finally yields
$$
\text{MSE}(\hat{\boldsymbol{\theta}}) \geq 
\frac{\sigma^2}{2}\text{Tr}\bigg[\frac{\partial \mathbf{h}}{\partial \boldsymbol{\theta}} \mathfrak{Re}\bigg\{\frac{\partial \mathbf{h}}{\partial \boldsymbol{\theta}}^H \mathbf{M}\mathbf{M}^H \frac{\partial \mathbf{h}}{\partial \boldsymbol{\theta}}\bigg\}^{-1}\frac{\partial \mathbf{h}}{\partial \boldsymbol{\theta}}^H\bigg].
$$}
\new{In this paper, we analyze the right-hand side of this inequality. It is expressed in a compact way with help of the introduced variation space in section~\ref{sec:CRB_variation_space}. Then, viewed as a function of the observation matrix $\mathbf{M}$, it is optimized under a power constraint in section~\ref{sec:optimCRB} in order to exhibit optimal pilot sequences and the associated minimal error, for \emph{any} deterministic channel model. Note that the approach we propose can be generalized to deal with improper measurements \cite{Darsena2004,Darsena2005,Abdallah2011}, using a more general form of the Slepian-Bangs formula \cite{Delmas2004}.}

\new{\section{CRB based on the variation space}
\label{sec:CRB_variation_space}}

In this section, the notion of variation space, which plays a central role in the analysis we propose, is first introduced and discussed. Then, the CRB is expressed as a function of the variation space.

\new{\subsection{Variation space and related notions}}
\label{sec:prelim}
\begin{definition}\emph{(Variation space)}
Let the set
$$
\mathcal{V}_{\boldsymbol{\theta}}\triangleq \bigg\{\frac{\partial\mathbf{h}}{\partial\boldsymbol{\theta}}\mathbf{x}, \,\mathbf{x} \in \mathbb{R}^{N_{p}}\bigg\}
$$ 
be the variation space around the parameters value $\boldsymbol{\theta}$.
This is the set corresponding to the potential directions of variation of the channel due to infinitesimal variations in the parameters value.
\label{def:variation_space}
\end{definition}
It is interesting to note that the variation space has the structure of an $\mathbb{R}$-vector space, since it contains all linear combinations of the columns of $\frac{\partial\mathbf{h}}{\partial\boldsymbol{\theta}}$ with \emph{real} coefficients. However, $\mathcal{V}_{\boldsymbol{\theta}}$ is not necessarily a $\mathbb{C}$-vector space, since it does not contain all linear combinations of the columns of $\frac{\partial\mathbf{h}}{\partial\boldsymbol{\theta}}$ with \emph{complex} coefficients (because we consider real parameters). This subtle distinction will play a major role in the subsequent analysis, as evidenced in section~\ref{sec:optimCRB}. To makes things clearer, we define also the real inner product $$\langle \mathbf{x},\mathbf{y} \rangle_{\mathbb{R}} \triangleq \mathfrak{Re}\{\mathbf{x}^H\mathbf{y}\}.$$ Two vectors $\mathbf{x}$ and $\mathbf{y}$ are said to be real-orthogonal (or $\mathbb{R}$-orthogonal) if $\langle \mathbf{x},\mathbf{y} \rangle_{\mathbb{R}}=0$. Let $\mathcal{E}$ be a $\mathbb{R}$-vector space, we denote $\text{dim}_{\mathbb{R}}(\mathcal{E})$ its dimension with the scalar field $\mathbb{R}$. Similarly, we denote the classical complex inner product $$\langle \mathbf{x},\mathbf{y} \rangle_{\mathbb{C}} \triangleq \mathbf{x}^H\mathbf{y},$$ and two vectors $\mathbf{x}$ and $\mathbf{y}$ are said to be complex-orthogonal ((or $\mathbb{C}$-orthogonal)) if $\langle \mathbf{x},\mathbf{y} \rangle_{\mathbb{C}}=0$. Let $\mathcal{F}$ be a $\mathbb{C}$-vector space, we denote $\text{dim}_{\mathbb{C}}(\mathcal{F})$ its dimension with the scalar field $\mathbb{C}$. Note that any $\mathbb{C}$-vector space is also a $\mathbb{R}$-vector space of doubled dimension, so that $\text{dim}_{\mathbb{R}}(\mathcal{F}) = 2\text{dim}_{\mathbb{C}}(\mathcal{F})$, but the converse is \emph{not} true (a $\mathbb{R}$-vector space is in general not a $\mathbb{C}$-vector space).

\subsection{Expression of the CRB}
\label{ssec:CRB}

\newnew{In the general setting we consider, the CRB can be expressed in a very simple way, depending only on the variation space, the observation matrix and the noise level. The obtained form of the CRB will prove very useful in section~\ref{sec:optimCRB} in order to optimize the observation matrix, and thus the sent pilot sequences. It is given by the following theorem.}
\begin{theorem} Provided $\text{dim}_\mathbb{R}(\mathcal{V}_{\boldsymbol{\theta}}) = N_p$, the Cram\'er-Rao bound is expressed as
$$\text{CRB}(\boldsymbol{\theta},\mathbf{M}) = \frac{\sigma^2}{2}\text{Tr}\left[\mathfrak{Re}\Big\{\mathbf{U}^H \mathbf{M}\mathbf{M}^H \mathbf{U}\Big\}^{-1} \right],$$
where $\mathbf{U}$ is any matrix whose columns form an $\mathbb{R}$-orthonormal basis of the variation space $\mathcal{V}_{\boldsymbol{\theta}}$.
\label{thm:CRB}
\end{theorem}

\begin{proof}

Let us start from \eqref{eq:CRB_start} and \eqref{eq:slepian}. In this basic form, the FIM is difficult to invert, because it involves the real part of a complex matrix. In \cite{Lemagoarou2018}, we proposed to use real representations of complex matrices to get rid of this problem. Here, in order to gain a deeper geometric understanding of the bound, let us use the Gram-Schmidt process on the gradient matrix $\frac{\partial \mathbf{h}}{\partial \boldsymbol{\theta}}$, with the real inner product $\langle .,. \rangle_{\mathbb{R}}$ to decompose it as
\begin{equation}
\frac{\partial \mathbf{h}}{\partial \boldsymbol{\theta}} = \mathbf{UR},
\label{eq:qr}
\end{equation}
where $\mathbf{U} \in \mathbb{C}^{N_d \times K}$ is a matrix whose columns are $\mathbb{R}$-orthonormal (meaning that $\mathfrak{Re}\left\{\mathbf{U}^H\mathbf{U}\right\} = \mathbf{Id}_{K}$), $\mathbf{R} \in \mathbb{R}^{K \times N_p}$ is a real upper-triangular matrix, and $K=\text{dim}_{\mathbb{R}}(\mathcal{V}_{\boldsymbol{\theta}})$. This decomposition of the gradient matrix allows to rewrite the FIM
$$
\mathbf{I}(\boldsymbol{\theta},\mathbf{M})= \frac{2}{\sigma^2}\mathbf{R}^T\mathfrak{Re}\Big\{\mathbf{U}^H \mathbf{M}\mathbf{M}^H \mathbf{U}\Big\}\mathbf{R},
$$
since $\mathfrak{Re}\{\mathbf{A}^H\mathbf{BA}\} = \mathbf{A}^T\mathfrak{Re}\{\mathbf{B}\}\mathbf{A}$ as soon as $\mathbf{A}$ is a real matrix.
Now, if and only if $\mathbf{R}$ is invertible, which is equivalent to $K=N_p$, the CRB is expressed
\begin{align*}
&\text{CRB}(\boldsymbol{\theta},\mathbf{M}) = \text{Tr}\bigg[\frac{\partial \mathbf{h}}{\partial \boldsymbol{\theta}} \mathbf{I}(\boldsymbol{\theta},\mathbf{M})^{-1}\frac{\partial \mathbf{h}}{\partial \boldsymbol{\theta}}^H\bigg] \\
&= \frac{\sigma^2}{2}\text{Tr}\left[ \mathbf{UR}\mathbf{R}^{-1}\mathfrak{Re}\Big\{\mathbf{U}^H \mathbf{M}\mathbf{M}^H \mathbf{U}\Big\}^{-1} \mathbf{R}^{-T}\mathbf{R}^{T}\mathbf{U}^H\right]\\
&= \frac{\sigma^2}{2}\text{Tr}\left[\mathfrak{Re}\Big\{\mathbf{U}^H \mathbf{M}\mathbf{M}^H \mathbf{U}\Big\}^{-1} \right] .
\end{align*}
In order to conclude, one can remark that 
\begin{equation*}
\resizebox{\hsize}{!}{$
\text{Tr}\left[\mathfrak{Re}\Big\{\mathbf{U}^H \mathbf{M}\mathbf{M}^H \mathbf{U}\Big\}^{-1} \right] = \text{Tr}\left[\mathfrak{Re}\Big\{\mathbf{B}^T\mathbf{U}^H \mathbf{M}\mathbf{M}^H \mathbf{U}\mathbf{B}\Big\}^{-1} \right]
$}
\end{equation*}
for any real orthogonal matrix $\mathbf{B} \in \mathbb{R}^{N_p\times N_p}$, so that the equation holds true for any matrix whose columns form an $\mathbb{R}$-orthogonal basis of $\mathcal{V}_{\boldsymbol{\theta}}$. 
\end{proof}

\newnew{This theorem allows to express the CRB in a way that is particularly suited to the determination of optimal observation matrices that is carried out in section~\ref{sec:optimCRB}.} Moreover, the CRB can be given an even simpler form. Indeed, it shows an invariance property, it is true for any matrix $\mathbf{U}$ whose columns are an $\mathbb{R}$-orthonormal basis of $\mathcal{V}_{\boldsymbol{\theta}}$. Moreover, the matrix $\mathfrak{Re}\big\{\mathbf{U}^H \mathbf{M}\mathbf{M}^H \mathbf{U}\big\}$ can be given a nice interpretation. Indeed, the orthogonal projection $\mathbf{P}_{\mathcal{V}_{\boldsymbol{\theta}}} \mathbf{z}$ of any vector $\mathbf{z}$ onto $\mathcal{V}_{\boldsymbol{\theta}}$ is expressed
$$
\mathbf{P}_{\mathcal{V}_{\boldsymbol{\theta}}} \mathbf{z} = \sum\nolimits_{i=1}^{N_p}\langle \mathbf{u}_i,\mathbf{z} \rangle_{\mathbb{R}} \mathbf{u}_i = \mathbf{U}\mathfrak{Re}\{\mathbf{U}^H\mathbf{z}\},
$$
so that $\mathfrak{Re}\{\mathbf{U}^H\mathbf{z}\}$ corresponds to the coordinates of the projection in the basis given by $\mathbf{U}$. Now, if $\mathbf{z} = \mathbf{M}\mathbf{M}^H\mathbf{t}$ with $\mathbf{t} \in \mathcal{V}_{\boldsymbol{\theta}}$ then 
$\mathfrak{Re}\big\{\mathbf{U}^H \mathbf{M}\mathbf{M}^H \mathbf{t}\big\} = \mathfrak{Re}\big\{\mathbf{U}^H \mathbf{M}\mathbf{M}^H \mathbf{U}\big\}\mathbf{r}$ for some $\mathbf{r} \in \mathbb{R}^{N_p}$ corresponding to the coordinates of $\mathbf{t}$ in the basis given by $\mathbf{U}$. It means that $\mathfrak{Re}\big\{\mathbf{U}^H \mathbf{M}\mathbf{M}^H \mathbf{U}\big\}$ is the matrix that corresponds to the operator $\mathbf{P}_{\mathcal{V}_{\boldsymbol{\theta}}}\mathbf{M}\mathbf{M}^H$ restricted to $\mathcal{V}_{\boldsymbol{\theta}}$ when expressed in the basis given by $\mathbf{U}$. Such an operator corresponds to the notion of compression in functional analysis. 

\begin{definition} \emph{(Compression \cite[p.120]{Halmos1982})}
Let $\mathcal{H}$ be a subspace of a Hilbert space $\mathcal{K}$, let $\mathbf{P}_{\mathcal{H}}$ be the orthogonal projection from $\mathcal{K}$ onto $\mathcal{H}$, and let $\mathbf{B}:\mathcal{K}\rightarrow \mathcal{K}$ be a linear operator on $\mathcal{K}$. The linear operator $\mathbf{A}:\mathcal{H}\rightarrow \mathcal{H}$ is the compression of $\mathbf{B}$ to $\mathcal{H}$, denoted $\left[ \mathbf{B}\right]_{\mathcal{H}}$, if
$$
\mathbf{Ax} = \mathbf{P}_{\mathcal{H}} \mathbf{Bx},\quad \forall \mathbf{x}\in \mathcal{H}.
$$
\label{def:compression}
\end{definition}
In the following, and when no confusion is possible, we denote the same way a matrix $\mathbf{A}$ and the operator associated to the multiplication by $\mathbf{A}$. Moreover, for an operator $\mathbf{A}:\mathcal{H}\rightarrow \mathcal{H}$ where $\mathcal{H}$ is a $\mathbb{K}$-vector space ($\mathbb{K} \in \{\mathbb{R},\mathbb{C}\}$), we define its trace as
$$
\text{Tr}\left[\mathbf{A}\right] \triangleq \sum\nolimits_{i=1}^{N_p}\langle \mathbf{v}_i,\mathbf{A}\mathbf{v}_i \rangle_{\mathbb{K}},
$$
where $\{\mathbf{v}_1,\dots,\mathbf{v}_{N_p}\}$ is any $\mathbb{K}$-orthonormal basis of $\mathcal{H}$. It coincides with the sum of the diagonal elements of a matrix when the operator action is a matrix multiplication. These two notions allow to express the CRB in a simpler and more intrinsic form, as in the following corollary (which is nothing more than a coordinate-free version of theorem~\ref{thm:CRB}).
\begin{corollary} Provided $\text{dim}_\mathbb{R}(\mathcal{V}_{\boldsymbol{\theta}}) = N_p$, the Cram\'er-Rao bound admits an intrinsic expression as
$$\text{CRB}(\boldsymbol{\theta},\mathbf{M}) 
 =  \frac{\sigma^2}{2} \text{Tr}\left[\left(\big[\mathbf{MM}^H\big]_{\mathcal{V}_{\boldsymbol{\theta}}}\right)^{-1}\right],$$
where $\big[\mathbf{MM}^H\big]_{\mathcal{V}_{\boldsymbol{\theta}}}$ is the compression of $\mathbf{MM}^H$ to the variation space $\mathcal{V}_{\boldsymbol{\theta}}$.
\label{cor:intrinsic_CRB}
\end{corollary}

 This form of the CRB shows that the minimal variance of any unbiased estimator is determined by the interaction between the observation matrix $\mathbf{M}$ and the potential directions of variations of the channel due to infinitesimal variations of the parameters around their value, represented by the set $\mathcal{V}_{\boldsymbol{\theta}}$. This fact, which is key in our analysis, is further exploited in the following section.

\new{\section{Optimized observation matrices}
\label{sec:optimCRB}}
\new{
In this section, the objective is to optimize the observation matrix $\mathbf{M}$ with respect to the particular form of the CRB given in theorem~\ref{thm:CRB}. We first give identifiability conditions, which allow to determine a minimal number of observations. Then, the optimal CRB and associated observation matrices of minimal size are given. Finally, we give a practical algorithm to design observation matrices based on an estimation of the variation space.
}

\subsection{Identifiability}
\label{ssec:identifiability}
Parameters are said to be identifiable if and only if the CRB is finite,
$$
\text{Identifiability} \Leftrightarrow\text{CRB}(\boldsymbol{\theta},\mathbf{M}) < +\infty.
$$
Identifiability imposes conditions on the variation space $\mathcal{V}_{\boldsymbol{\theta}}$ and on the observation matrix $\mathbf{M}$, as  
stated in the following theorem.

\begin{theorem} The parameters are identifiable if and only if 
$$
\text{dim}_\mathbb{R}(\mathcal{V}_{\boldsymbol{\theta}}) = N_p
$$
and
$$
\mathcal{V}_{\boldsymbol{\theta}} \cap \text{im}_\mathbb{C}(\mathbf{M})^\perp = \{\mathbf{0}\}.
$$
\label{thm:identif}
\end{theorem}
\begin{proof}
The first condition $\text{dim}_\mathbb{R}(\mathcal{V}_{\boldsymbol{\theta}}) = N_p$ is equivalent to the invertibility of $\mathbf{R}$ that was shown to be a necessary condition for the CRB to be finite in section~\ref{ssec:CRB}.
When this condition is fulfilled, identifiability holds if and only if the matrix $\mathfrak{Re}\left\{\mathbf{U}^H\mathbf{MM}^H\mathbf{U} \right\}$ is invertible. This matrix being symmetric, it is invertible if and only if
$$
\forall \mathbf{x}\neq \mathbf{0} \in \mathbb{R}^{N_p},\, \mathbf{x}^T\mathfrak{Re}\left\{\mathbf{U}^H\mathbf{MM}^H\mathbf{U}\right\}\mathbf{x} \neq 0 .
$$ 
Moreover, for any real vector $\mathbf{x}$, $\mathbf{x}^T\mathfrak{Re}\left\{\mathbf{U}^H\mathbf{MM}^H\mathbf{U}\right\}\mathbf{x} = \mathbf{x}^T\mathbf{U}^H\mathbf{MM}^H\mathbf{U}\mathbf{x}$. Thus, recalling that $\mathcal{V}_{\boldsymbol{\theta}} = \text{im}_{\mathbb{R}}(\mathbf{U})$, identifiability holds if and only if
$$
\forall \mathbf{z}\neq \mathbf{0} \in \mathcal{V}_{\boldsymbol{\theta}},\, \mathbf{z}^H\mathbf{MM}^H\mathbf{z} = \left\Vert \mathbf{M}^H\mathbf{z} \right\Vert_2^2 \neq 0,
$$
which is equivalent (since $\text{ker}(\mathbf{M}^H) = \text{im}_{\mathbb{C}}(\mathbf{M})^\perp$) to 
$$
\mathcal{V}_{\boldsymbol{\theta}} \cap \text{im}_\mathbb{C}(\mathbf{M})^\perp = \{\mathbf{0}\}.
$$
\end{proof}


\noindent {\bf Interpretations.} The first identifiability condition $\text{dim}_\mathbb{R}(\mathcal{V}_{\boldsymbol{\theta}}) = N_p$ means that the columns of $\frac{\partial \mathbf{h}}{\partial \boldsymbol{\theta}}$ have to be linearly independent over $\mathbb{R}$ for identifiability to be possible, whatever the observation matrix. Said differently, the number  of degrees of freedom of the variation space has to be equal to the number of parameters to estimate, so that small variations of the channel $\mathbf{h}$ due to an infinitesimal variation in  the value of any parameter cannot be mistaken with small variations of the channel due to infinitesimal variations in the values of the other parameters. Note that since $\text{dim}_\mathbb{R}(\mathcal{V}_{\boldsymbol{\theta}}) \leq 2N_d$, this condition implies $N_p \leq 2N_d$, which means that it is impossible to identify a number of parameters that is more than twice the dimension of the channel.

Then, if the first condition is fulfilled, the second condition $\mathcal{V}_{\boldsymbol{\theta}} \cap \text{im}_\mathbb{C}(\mathbf{M})^\perp = \{\mathbf{0}\}$ means that no nonzero vector in the space of variations $\mathcal{V}_{\boldsymbol{\theta}}$ can be orthogonal to the column space of the observation matrix $\mathbf{M}$ for identifiability to hold. Said differently, the observation matrix has to preserve some energy for any element of the space of variations, every infinitesimal variation in the values of the parameters has to cause a change in the observation vector $\mathbf{y}$.

\noindent {\bf Number of observations.} Identifiability directly imposes a minimal number of observations $N_m$, as stated in the following corollary.
\begin{corollary}
Parameters can be identifiable only if 
$$N_m\geq  \frac{N_p}{2}.$$
\label{cor:minobs}
\end{corollary}
\begin{proof} Identifiability can be stated:
$$
\forall \mathbf{z}\neq\mathbf{0}\in \mathcal{V}_{\boldsymbol{\theta}},\, \mathbf{M}^H\mathbf{z} \neq \mathbf{0},
$$ 
which is possible only if the $\mathbb{R}$-dimension of $\text{ker}(\mathbf{M}^H)$ plus the $\mathbb{R}$-dimension of $\mathcal{V}_{\boldsymbol{\theta}}$ is no greater than the $\mathbb{R}$-dimension of the ambient space $\mathbb{C}^{N_d}$ (so that they can have a trivial intersection). This writes
$$
\text{dim}_\mathbb{R}(\text{ker}(\mathbf{M}^H)) + N_p \leq 2N_d.
$$ 
Moreover, $\text{dim}_\mathbb{R}(\text{ker}(\mathbf{M}^H)) = 2N_d - \text{dim}_\mathbb{R}(\text{im}_\mathbb{C}(\mathbf{M}^H))$ (rank-nullity theorem), so that we end up with
$$
\text{dim}_\mathbb{R}(\text{im}_\mathbb{C}(\mathbf{M}^H)) \geq N_p.
$$
The $\mathbb{R}$-dimension of a $\mathbb{C}$-vector space being twice its $\mathbb{C}$-dimension and the $\mathbb{C}$-dimension being upper-bounded by the number of columns, we finally get
$$
N_m \geq \text{dim}_\mathbb{C}(\text{im}_\mathbb{C}(\mathbf{M}^H)) \geq \frac{N_p}{2},
$$
which proves the result.
\end{proof}
We just showed that the minimal number of observations $N_m$ required for identifiability to be possible is $\lceil \frac{N_p}{2} \rceil$. In other words, the matrix $\mathbf{M}$ has to have at least $\lceil \frac{N_p}{2} \rceil$ columns for the CRB to be finite. As will be shown in the next subsection, there always exist an optimal observation matrix having $\lceil \frac{N_p}{2} \rceil$ columns.

\subsection{Optimality}
\label{ssec:optimality}
Let us now determine the minimal value of the CRB under a power constraint, and the observation matrices allowing to attain it. This corresponds to solve the optimization problem: 
\begin{align}
\begin{split}
\underset{\mathbf{M}}{\text{minimize}} &\quad \text{CRB}(\boldsymbol{\theta},\mathbf{M}), \\
\text{subject to} &\quad \left\Vert \mathbf{M} \right\Vert_F^2=P.
\end{split}
\label{eq:pb_optim}
\end{align}
Note that the quantity $\left\Vert \mathbf{M} \right\Vert_F^2 = P = \text{Tr}(\mathbf{MM}^H)$ corresponds to the observation power, which is proportional to the received power and not directly equal to the transmitted power. The two quantities are linked in section \ref{sec:illustration}.

\moved{
\subsubsection{Decomposition of the variation space}
\label{ssec:var_space}
\new{The expression of the CRB given in theorem~\ref{thm:CRB} is valid for any $\mathbb{R}$-orthogonal basis of $\mathcal{V}_{\boldsymbol{\theta}}$. In order to ease optimization, we exhibit here a specific basis with useful properties.}
}
\moved{
To do so, let us state the following lemma that allows to decompose $\mathcal{V}_{\boldsymbol{\theta}}$ into a direct sum of $\mathbb{C}$-orthogonal subspaces.
\begin{lemma}
(i) Any $\mathbb{R}$-vector space $\mathcal{E}$ of dimension $d$ that belongs to a $\mathbb{C}$-vector space $\mathcal{F}$ (containing $\mathrm{j}\mathcal{E}$) can be decomposed into the direct sum of subspaces of dimension $2$ (and possibly a subspace of dimension one if $d$ is odd) that are mutually $\mathbb{C}$-orthogonal. (ii) The subspaces of the aforementioned decomposition belong to eigenspaces of $\mathbf{P}_{\mathcal{E}}\circ\mathbf{P}_{\mathrm{j}\mathcal{E}}$, where $\mathbf{P}_{\mathcal{E}}$ (resp.\ $\mathbf{P}_{\mathrm{j}\mathcal{E}}$) is the orthogonal projection onto $\mathcal{E}$ (resp.\ $\mathrm{j}\mathcal{E}$).
\label{lem:decomp_V_theta}
\end{lemma}
\begin{proof}
\new{This lemma is proven in appendix~\ref{proof:lemma1}.}
\end{proof}
}

\moved{
Applying lemma~\ref{lem:decomp_V_theta} to the variation space $\mathcal{V}_{\boldsymbol{\theta}}$ (assuming it is of dimension $N_p$ and $N_p$ is even), it is possible to decompose it as
\begin{equation}
\mathcal{V}_{\boldsymbol{\theta}} = \text{span}_{\mathbb{R}}\left( \left\{\mathbf{v}_1,\mathbf{w}_1,\dots,\mathbf{v}_{\frac{N_p}{2}},\mathbf{w}_{\frac{N_p}{2}} \right\}\right)
\label{eq:decomp_V_theta}
\end{equation}
where $\mathbf{v}_m^H\mathbf{v}_n = \delta_{mn}$, $\mathbf{w}_m^H\mathbf{w}_n = \delta_{mn}$ and $\mathbf{v}_m^H\mathbf{w}_n = -\delta_{mn}\mathrm{j}c_m$ (\new{with $0\leq c_m \leq 1$}, and $\delta$ being the Kronecker symbol). \new{The quantities $c_m$ can be seen as the lack of $\mathbb{C}$-orthogonality of the $\mathbb{R}$-orthogonal basis $\left\{\mathbf{v}_1,\mathbf{w}_1,\dots,\mathbf{v}_{\frac{N_p}{2}},\mathbf{w}_{\frac{N_p}{2}} \right\}$.} Let us introduce the matrix
\begin{equation}
\mathbf{V}\triangleq \left(\mathbf{v}_1,\mathbf{w}_1,\dots,\mathbf{v}_{\frac{N_p}{2}},\mathbf{w}_{\frac{N_p}{2}} \right)
\label{eq:matrix_V}
\end{equation}
whose columns form an $\mathbb{R}$-orthonormal basis of $\mathcal{V}_{\boldsymbol{\theta}}$.
 Similarly, if $N_p$ is odd, the decomposition reads 
\begin{equation}
\mathcal{V}_{\boldsymbol{\theta}} = \text{span}_{\mathbb{R}}\left( \left\{\mathbf{v}_1,\mathbf{w}_1,\dots,\mathbf{v}_{\left\lfloor\frac{N_p}{2}\right\rfloor},\mathbf{w}_{\left\lfloor\frac{N_p}{2}\right\rfloor},\mathbf{v}_{\left\lfloor\frac{N_p}{2}\right\rfloor+1}\right\}\right),
\label{eq:decomp_V_theta_odd}
\end{equation}
where $\mathbf{v}_m^H\mathbf{v}_n = \delta_{mn}$, $\mathbf{w}_m^H\mathbf{w}_n = \delta_{mn}$ and $\mathbf{v}_m^H\mathbf{w}_n = -\delta_{mn}\mathrm{j}c_m$, and the matrix $\mathbf{V}$ can be built the same way. 
\new{Said differently, this result means that for any matrix $\mathbf{U}$ whose columns form an $\mathbb{R}$-orthogonal basis of $\mathcal{V}_{\boldsymbol{\theta}}$, there exists a real orthogonal matrix $\mathbf{B}$ such that
\begin{equation}
\mathbf{B}^T\mathfrak{Im}\{\mathbf{U}^H\mathbf{U}\}\mathbf{B} = 
\begin{pmatrix}
0&-c_1 & \\
c_1 & 0 & \\
& & 0&-c_2 \\
& & c_2& 0 \\ 
& & & &\ddots \\
\end{pmatrix} \triangleq \boldsymbol{\Gamma},
\label{eq:youla} 
\end{equation}
and then $\mathbf{V} = \mathbf{UB}$. In practice, the matrices $\mathbf{B}$ and $\boldsymbol{\Gamma}$ can be obtained by computing the real Schur decomposition of the matrix $\mathfrak{Im}\{\mathbf{U}^H\mathbf{U}\}$ and reordering the diagonal blocks.
}
}


\subsubsection{Optimal CRB and observation matrices}Using this decomposition allows us to state the main result of this paper in the following theorem.
\begin{theorem}
The minimal value of the CRB is
$$
\text{CRB}_\text{min}(\boldsymbol{\theta}) \triangleq \frac{2\sigma^2}{P}\left( \sum\nolimits_{k=1}^{\lfloor\frac{N_p}{2}\rfloor} \frac{1}{\sqrt{1+c_k}}  + \frac{\epsilon}{2}\right)^2,
$$
where the scalars $c_k$ are defined at \eqref{eq:decomp_V_theta} (resp. \eqref{eq:decomp_V_theta_odd}) and $\epsilon=0$ (resp. $\epsilon=1$) if $N_p$ is even (resp. odd).

It is attained with the observation matrix of minimal size
$$
\mathbf{M} = \sqrt{\frac{P}{C}}   \left( \frac{\mathbf{v}_1 + \mathrm{j}\mathbf{w}_1}{(1+c_1)^{\frac{3}{4}}},\dots, \frac{\mathbf{v}_{\frac{N_p}{2}} + \mathrm{j}\mathbf{w}_{\frac{N_p}{2}}}{(1+c_{\frac{N_p}{2}})^{\frac{3}{4}}} \right),
$$
where $C\triangleq 2\sum_{l=1}^{\frac{N_p}{2}}\frac{1}{\sqrt{1+c_l}}$ and the vectors $\mathbf{v}_k, \mathbf{w}_k$ are defined at \eqref{eq:decomp_V_theta} if $N_p$ is even, and with
$$
\mathbf{M} = \sqrt{\frac{P}{C}}   \left( \frac{\mathbf{v}_1 + \mathrm{j}\mathbf{w}_1}{(1+c_1)^{\frac{3}{4}}},\dots, \frac{\mathbf{v}_{\left\lfloor \frac{N_p}{2}\right\rfloor} + \mathrm{j}\mathbf{w}_{\left\lfloor \frac{N_p}{2}\right\rfloor}}{(1+c_{\left\lfloor \frac{N_p}{2}\right\rfloor})^{\frac{3}{4}}}, \mathbf{v}_{\left\lfloor \frac{N_p}{2}\right\rfloor + 1} \right),
$$
where $C\triangleq 2\sum_{l=1}^{\left\lfloor \frac{N_p}{2}\right\rfloor}\frac{1}{\sqrt{1+c_l} }+ 1$ and the vectors $\mathbf{v}_k, \mathbf{w}_k$ are defined at \eqref{eq:decomp_V_theta_odd} if $N_p$ is odd.
\label{thm:optim}
\end{theorem}

\begin{proof}
\new{This theorem is proven in appendix~\ref{proof:thm3}.}
\end{proof}

This theorem exhibits the fact that the optimal CRB depends on the noise level $\sigma^2$, the observation power $P$ and the properties of the variation space $\mathcal{V}_{\boldsymbol{\theta}}$, namely its dimension $N_p$ and the quantities $c_k$. Moreover it can be bounded above and below as
\begin{equation}
\frac{\sigma^2N_p^2}{4P} \leq \text{CRB}_\text{min}(\boldsymbol{\theta}) \leq \frac{\sigma^2N_p^2}{2P},
\label{eq:bound_CRB}
\end{equation}
with an equality on the left if and only if $N_p$ is even and $c_k=1$, $\forall k$ ($\mathcal{V}_{\boldsymbol{\theta}}$ is then a $\mathbb{C}$-vector space), and equality on the right if and only if $c_k=0$, $\forall k$ ($\mathcal{V}_{\boldsymbol{\theta}}$ is then $\mathbb{R}$-orthogonal to $\mathrm{j}\mathcal{V}_{\boldsymbol{\theta}}$). 

\new{\subsection{Observation matrix design}
\label{ssec:design}}
\new{
Based on theorem~\ref{thm:optim}, it is possible to build optimal observation matrices of minimal size $\lceil \frac{N_p}{2} \rceil$, provided the variation space is known. Since the variation space depends itself on the parameters to estimate, this result may seem of little use. However, in some cases, an estimation $\hat{\mathcal{V}}_{\boldsymbol{\theta}}$ of the variation space can be obtained. For example, this is the case for MIMO systems operating in FDD mode using a physical model, where the variation space can be determined based on the previous uplink or downlink channel estimates since it depends only on the directions of arrival of the channel paths, which vary in general slowly. This interesting application is studied in details in the next section.\begin{algorithm}[htb]
\new{\caption{\new{Observation matrix determination ($N_p$ even)}}
\begin{algorithmic}[1] 
\REQUIRE An estimate $\hat{\mathcal{V}}_{\boldsymbol{\theta}}$ of the variation space (i.e. a matrix $\mathbf{G}$ whose columns are a generating family of $\hat{\mathcal{V}}_{\boldsymbol{\theta}}$ with real scalars), the observation power $P$.
\STATE Find a matrix $\mathbf{U}$ whose columns form an $\mathbb{R}$-orthonormal basis of $\hat{\mathcal{V}}_{\boldsymbol{\theta}}$: \\  Real QR decomposition \eqref{eq:qr}: $\begin{pmatrix}
\mathfrak{Re}\{\mathbf{G}\} \\
\mathfrak{Im}\{\mathbf{G}\}
\end{pmatrix} =  \mathbf{QR}$, \\
$\mathbf{U} \leftarrow \frac{\widehat{\partial \mathbf{h}}}{\partial \boldsymbol{\theta}}\mathbf{R}^{-1}$
\STATE Decompose $\hat{\mathcal{V}}_{\boldsymbol{\theta}}$ as in lemma~\ref{lem:decomp_V_theta}:\\ Real Schur decomposition \eqref{eq:youla}: $ \mathfrak{Im}\{\mathbf{U}^H\mathbf{U}\} = \mathbf{B}\boldsymbol{\Gamma}\mathbf{B}^T$,  \\
$c_1 \leftarrow \gamma_{21}$, $c_2 \leftarrow \gamma_{43}$,$\dots$,$c_{\frac{N_p}{2}} \leftarrow \gamma_{N_p,N_p-1}$,\\ $\mathbf{V} \leftarrow \mathbf{UB}$
\STATE Build the observation matrix according to theorem~\ref{thm:optim}: \\~\\ $\mathbf{S} \leftarrow \left(\begin{matrix}
1 & 0 &  \dots \\
\mathrm{j} &0 &\dots \\
0 & 1 &  \dots \\
0 & \mathrm{j}  &\dots \\
\vdots &\vdots  &\ddots \\
\end{matrix}\right) \in \mathbb{C}^{N_p \times \frac{N_p}{2}} $,
 $C \leftarrow 2\sum_{l=1}^{\frac{N_p}{2}}\frac{1}{\sqrt{1+c_l}}$, \\ ~\\~\\
$\mathbf{D} \leftarrow \sqrt{\frac{P}{C}}\left(\begin{smallmatrix}
\frac{1}{(1+c_1)^{\frac{3}{4}}} & 0 &  \\
0 &\frac{1}{(1+c_2)^{\frac{3}{4}}} &  \\
 & & \ddots &
\end{smallmatrix}\right) \in \mathbb{C}^{\frac{N_p}{2} \times \frac{N_p}{2}} $, \\
$\mathbf{M} \leftarrow \mathbf{VSD}$
\ENSURE The observation matrix $\mathbf{M} \in \mathbb{C}^{N_d \times \frac{N_p}{2}}$ that is optimal with respect to $\hat{\mathcal{V}}_{\boldsymbol{\theta}}$.
\end{algorithmic}
\label{alg:single}}
\end{algorithm}
}

\new{The strategy we propose to build observation matrices based on an estimate of the variation space is given in algorithm~\ref{alg:single} for an even number of parameters (the algorithm is almost the same for an odd number, except for the third step being slightly modified according to theorem~\ref{thm:optim}). It comprises three steps. The first one amounts to find an $\mathbb{R}$-orthogonal basis of $\hat{\mathcal{V}}_{\boldsymbol{\theta}}$. The second one corresponds to apply the decomposition of lemma~\ref{lem:decomp_V_theta} to $\hat{\mathcal{V}}_{\boldsymbol{\theta}}$. Finally the third one uses the result of theorem~\ref{thm:optim} based on the aforementioned decomposition to build the observation matrix. Overall, the computational complexity of this algorithm is $\mathcal{O}(N_dN_p^2)$.
}

\section{Illustrations of the results}
\label{sec:illustration}
\new{
Let us now illustrate the applicability of the presented theoretical results, whose ultimate goal is to facilitate the design of optimal and short pilot sequences for any deterministic channel model. To do so, we consider a massive MIMO system and compare various models. 
\newline {\bf Scenarios.} Three scenarios are chosen to apply our results:
\begin{itemize}[leftmargin=*]
\item First, in section~\ref{ssec:LS}, the classical least squares model is studied in the proposed framework, showing that our approach is trivial in that case and allows to retrieve previous results. This model makes no use of any a priori information, so that pilot sequences have to be as long as the channel dimension and the optimal CRB is proportional to the square of the channel dimension. 
\item Then, in section~\ref{ssec:CP}, physical models are investigated, which allow to design shorter pilot sequences and theoretically lead to better performance, thanks to the lower number of parameters to estimate. In that case, our framework allows to theoretically justify previous approaches based on estimates of the channel directions of arrival (DoA), assuming their reciprocity \cite{Xie2017} or time persistence \cite{Gao2015}. These approaches lead to the length of optimal pilot sequences being proportional to the number of dominant paths and the optimal CRB being proportional to the square of this quantity. However, such methods induce a biased model, due to the fact that DoA estimates are considered perfect and kept fixed, \newnew{which lacks robustness in practical scenarios}. 
\item Finally, in section~\ref{ssec:newstrat}, still in the physical models context, we show that our result allows to suggest a new strategy \newnew{that is more robust to} the DoA estimation error. It is based on an update of the DoA estimates, and comes with better theoretical guarantees than previous approaches. Indeed, it corrects their bias and causes only a small increase of the optimal CRB and pilot sequence length (due to the DoA update). This approach is compared numerically to the one of section~\ref{ssec:CP}, taking into account the DoA estimation error, showing \newnew{empirically} its advantage \newnew{ in terms of robustness}.
\end{itemize}
}

\noindent {\bf Setting.} In the general case where the channel to estimate is between $N_t$ transmit antennas and $N_r$ receive antennas, on $N_f$ subcarriers, the channel is a complex vector of dimension $N_d = N_rN_tN_f$ denoted $\mathbf{h} \in \mathbb{C}^{N_rN_tN_f}$, where $h_{ijk}$ is the channel between the $j$-th transmit antenna and the $i$-th receive antenna on the $k$-th subcarriers. 
The observation matrix $\mathbf{M}$ takes a particular form in this context, \new{and can be linked directly to the sent pilot sequence}. Indeed, if the transmitter sends a pilot sequence of length $T$ corresponding to the matrix $\mathbf{X} \in \mathbb{C}^{N_t \times T}$ on $N_{\text{ps}}$ pilot subcarriers, then the signal at the receive antennas can be written as in \eqref{eq:observation} with
\begin{equation}
\mathbf{M}  = \mathbf{Id}_{N_r} \otimes \mathbf{X} \otimes \mathbf{F} \in \mathbb{C}^{N_rN_tN_f \times N_rTN_\text{ps}},
\label{eq:observation_matrix}
\end{equation}
where $\mathbf{F} \in \{0,1\}^{N_f \times N_\text{ps}}$ is a column-sampled identity matrix, keeping only the columns corresponding to the selected pilot subcarriers. In such a setting, the number of complex observations is $N_m = N_rTN_\text{ps}$, and the transmitted power is $P_t \triangleq N_\text{ps} \left\Vert \mathbf{X} \right\Vert_2^2$. On the other hand, the observation power which is constrained in the optimization problem~\eqref{eq:pb_optim}, is expressed  $P =  \left\Vert \mathbf{M} \right\Vert_2^2 = N_rN_{\text{ps}}\left\Vert \mathbf{X} \right\Vert_2^2$. We thus have $P = N_rP_t$, acknowledging the fact that adding receive antennas increases the received power $P$ without changing the transmitted power $P_t$. 

\new{
In this section, let us perform the analysis considering a massive MIMO setting where the base station is equipped with an uniform linear array (ULA) with half-wavelength separated antennas aligned with the $y$-axis, and user terminals are equipped with a single antenna ($N_r=1$). 
Let us also consider a single subcarrier ($N_f=1$) for ease of exposition, but note that the study straightforwardly extends to the multi-carrier case. These assumptions directly imply the direct equality of the observation matrix and the pilot sequences matrix: 
\begin{equation}
\mathbf{M} = \mathbf{X},\, P=P_t,
\label{eq:setting}
\end{equation}
which greatly simplifies the following subsections.}

\subsection{Application to the least squares model}
\label{ssec:LS}
The most direct and simple way to parameterize the channel with no a priori is to take as $N_p = 2N_t = 2N_d$ parameters the real and imaginary parts of the channel entries, 
$$
\boldsymbol{\theta}_{\text{LS}} \triangleq (\mathfrak{Re}(\mathbf{h})^T, \mathfrak{Im}(\mathbf{h})^T)^T \in \mathbb{R}^{2N_t},
$$ 
This leads to a linear channel model, expressed in function of the parameters as 
\begin{equation}
\mathbf{h}_{\text{LS}}(\boldsymbol{\theta}_{\text{LS}}) = \begin{pmatrix}
\mathbf{Id},\mathrm{j}\mathbf{Id}
\end{pmatrix}\boldsymbol{\theta}_{\text{LS}}.
\label{eq:LS_channel}
\end{equation}
The observation defined in \eqref{eq:observation} then reads $\mathbf{y} = (\mathbf{M}^H,\mathrm{j}\mathbf{M}^H)\boldsymbol{\theta}_{\text{LS}} + \mathbf{n}$ and the maximum likelihood estimation problem becomes a least squares problem, hence the name of the model. In this case, $\frac{\partial \mathbf{h}}{\partial\boldsymbol{\theta}_{\text{LS}} } = (\mathbf{Id},\mathrm{j}\mathbf{Id})$, which, following the definition of the variation space gives
\begin{equation}
\mathcal{V}_{\boldsymbol{\theta}_{\text{LS}}} =  \mathbb{C}^{N_t}.
\label{eq:v_theta_LS}
\end{equation}
\new{Regarding the framework proposed in this paper, this is a trivial case since the variation space is independent of the parameters value, due to the linearity of the model. Consequently, no estimation of the variation space is needed to design optimal pilot sequences.} Indeed, this particular variation space can decomposed according to lemma~\ref{lem:decomp_V_theta} as
$$
\mathcal{V}_{\boldsymbol{\theta}_{\text{LS}}} = \text{span}_{\mathbb{R}}\left(\left\{\mathbf{b}_1,-\mathrm{j}\mathbf{b}_1,\dots,\mathbf{b}_{N_t},-\mathrm{j}\mathbf{b}_{N_t}\right\}\right),
$$
where $\{\mathbf{b}_1,\dots,\mathbf{b}_{N_t}\}$ is any $\mathbb{C}$-orthonormal basis of $\mathbb{C}^{N_t}$, and $c_1=\dots = c_{N_t}=1$. 

\noindent {\bf Optimal CRB.} Applying theorem~\ref{thm:optim}, the optimal CRB of this model is then
\begin{equation}
\text{CRB}_\text{min}(\boldsymbol{\theta}_{\text{LS}}) = \frac{\sigma^2N_t^2}{P_t}.
\label{eq:optim_LS}
\end{equation}
It is attained for observation matrices (pilot sequences) of the form
\begin{equation}
\mathbf{X}_{\text{LS}} = \mathbf{M}_{\text{LS}} =  \sqrt{\frac{P}{N_t}}\left(\mathbf{b}_1,\dots,\mathbf{b}_{N_t} \right).
\label{eq:seq_LS}
\end{equation}
Such matrices have $N_t = \frac{N_p}{2}$ columns, which, according to corollary~\ref{cor:minobs} is minimal for identifiability to be possible.

 Equations \eqref{eq:optim_LS} and \eqref{eq:seq_LS} are nothing but a restatement of a well-known result \cite{Biguesh2006}. The method we propose here indeed allows to derive results for any linear channel model. However, it is more powerful since it generalizes also to nonlinear models, as shown in the next subsections. 

\subsection{Application to physical models}
\label{ssec:CP}

Another way to parameterize the channel is to assume that $\mathbf{h}$ is the sum of $L$ atomic channels corresponding to distinct physical paths, characterized by their direction of departure (DoD) $\overrightarrow{u_{t}}$, direction of arrival (DoA) $\overrightarrow{u_{r}}$, delay $\tau$ and complex gain $\beta$. Parameters of this kind of model are thus given by 
\begin{equation}
 \boldsymbol{\theta} = \left[\left(\mathfrak{Re}(\beta_l),\mathfrak{Im}(\beta_l),\overrightarrow{u_{r,l}},\overrightarrow{u_{t,l}},\tau_l\right)_{l=1}^L\right]^T.
\label{eq:physical_params}
\end{equation} 
The corresponding number of parameters is $N_p = mL$, where $m$ is the number of real parameters to be estimated per physical path. The quantity $m$ can take different values depending on the considered setting.

\new{In the massive MIMO setting studied here, \eqref{eq:physical_params} reduces to
\begin{equation}
 \boldsymbol{\theta}_{\text{PHY}} = \left[\left(\mathfrak{Re}(\beta_l),\mathfrak{Im}(\beta_l),\phi_l\right)_{l=1}^L\right]^T,
\label{eq:physical_params_new}
\end{equation}
where $\phi_l$ is the azimuth angle for the $l$-th physical path.
}
\new{The downlink channel is then expressed
\begin{equation}
\mathbf{h}_{\text{PHY}}(\boldsymbol{\theta}_{\text{PHY}}) = \sum_{l=1}^L\beta_l \mathbf{e}(\phi_l),
\label{eq:physical_channel}
\end{equation}
where $\mathbf{e}(\phi) = \frac{1}{\sqrt{N_t}}(\mathrm{e}^{-\mathrm{j}\frac{2\pi}{\lambda}\left(\frac{N_t-1}{4}\right)\sin\phi},\dots,\mathrm{e}^{\mathrm{j}\frac{2\pi}{\lambda}\left(\frac{N_t-1}{4}\right)\sin\phi})^T$ is the steering vector associated with azimuth $\phi$ (for an even number of antennas). This model is nonlinear (although it is linear with respect to the complex gains $\beta_l$). Such physical models are quite standard \cite{Sayeed2002,Bajwa2010} and successful due to the possibility to take $L$ small (typically less than ten) with good channel modeling accuracy (the channel is then called sparse). The variation space for such physical model is expressed
\begin{equation}
\mathcal{V}_{\boldsymbol{\theta}_{\text{PHY}}}  = \text{span}_\mathbb{R} \left(\left\{\mathbf{e}(\phi_l),-\mathrm{j}\mathbf{e}(\phi_l),\frac{\partial\mathbf{e}(\phi_l)}{\partial \phi_l}\right\}_{l=1}^L\right).
\label{eq:variation_phy}
\end{equation}
As opposed to the least squares case, this variation space depends on the parameter value, so that it is not possible to build optimal short-length pilot sequences without relying on some prior information about the parameters. However, it is interesting to notice that the variation space depends only on the azimuth angles $\phi_1,\dots,\phi_L$ and not on the path gains. Fortunately, the azimuth angles vary much more slowly than the path gains and are the same for the uplink and downlink channels, so that one can hope to acquire good estimates of them in order to design pilot sequences.
}
\new{\newline {\bf Angle-constrained estimation strategy.}}
Under these assumptions, it has already been proposed to design short length donwlink pilot sequences considering azimuth angles are already known by the base station, either thanks to uplink channel estimates \cite{Xie2017} (using the channel angle reciprocity), or thanks to previous downlink channel estimates \cite{Gao2015} (using angle time persistence).
\newnew{
In both cases, the communication workflow can be decomposed in two separate phases as follows:
\begin{enumerate}
\item {\bf Paths estimation:} estimate all dominant path angles, either sending long downlink pilot sequences \cite{Gao2015} (duration at least equal to the number of transmit antennas) or using the uplink pilots \cite{Xie2017}.
\item {\bf Channel tracking:} track the downlink channel, by estimating path gains while keeping path angles fixed, sending short downlink pilot sequences (duration proportional to the number of dominant paths). 
\end{enumerate}
Note that the paths estimation phase (step 1) has to be done periodically, in order to take into account changes in path angles, the disappearance of old paths and the appearance of new paths. Otherwise, the link quality may degrade much due to the aging of the path angles information.}

 Let us analyze these methods that we call angle-constrained within our framework. \newnew{We focus here on the channel tracking phase.} Denoting $\hat{\phi}_1,\dots,\hat{\phi}_L$ the estimated azimuths and  $\hat{\mathbf{E}} \in \mathbb{C}^{N_d\times L} \triangleq (\mathbf{e}(\hat{\phi}_1),\dots,\mathbf{e}(\hat{\phi}_L))$ the matrix of corresponding steering vectors, they amount to simplify the estimation problem considering only the path gains remain to estimate, yielding parameters
\begin{equation}
 \boldsymbol{\theta}_{\text{AC}} \triangleq \left[\left(\mathfrak{Re}(\beta_l)_{l=1}^L,\mathfrak{Im}(\beta_l)_{l=1}^L\right)\right]^T.
\label{eq:CP_params}
\end{equation}
It thus corresponds to have $N_p = 2L$, and the channel is expressed 
\begin{equation}
\mathbf{h}_{\text{AC}}(\boldsymbol{\theta}_{\text{AC}}) = \begin{pmatrix}
\hat{\mathbf{E}},\mathrm{j}\hat{\mathbf{E}}
\end{pmatrix}\boldsymbol{\theta}_{\text{AC}}.
\label{eq:CP_channel}
\end{equation}
Considering this model, the variation space simplifies to 
$$
\mathcal{V}_{\boldsymbol{\theta}_{\text{AC}}} = \text{span}_{\mathbb{R}}\left(\left\{\mathbf{e}(\hat{\phi}_1),-\mathrm{j}\mathbf{e}(\hat{\phi}_1),\dots,\mathbf{e}(\hat{\phi}_L),-\mathrm{j}\mathbf{e}(\hat{\phi}_L)\right\}\right).
$$
\newline {\bf Optimal CRB.} Thus, applying theorem~\ref{thm:optim} in that particular case, the optimal CRB of the angle-constrained physical model is bounded as
\begin{equation}
\frac{\sigma^2L^2}{P_t}\leq \text{CRB}_\text{min}(\boldsymbol{\theta}_{\text{PHY}}) \leq 2\times \frac{\sigma^2L^2}{P_t},
\label{eq:optim_LP}
\end{equation}
where the lower bound is attained if the columns of $\hat{\mathbf{E}}$ are mutually orthogonal (which is assumed in \cite{Gao2015,Xie2017}), in which case the decomposition of lemma~\ref{lem:decomp_V_theta} is trivial with $c_1 = \dots = c_L = 1$.
In that case, it is interesting to compare this CRB to the CRB of the least squares model \eqref{eq:optim_LS}: using the physical model allows to divide by $\left(\frac{N_t}{L}\right)^2$ the minimal attainable variance. This potentially huge gain is attained provided the azimuth estimates are perfect, and for downlink pilot sequences of the form
\begin{equation}
\mathbf{X}_{\text{AC}}  = \mathbf{M}_{\text{AC}} =  \sqrt{\frac{P_t}{L}}\left(\mathbf{e}(\hat{\phi}_1),\dots,\mathbf{e}(\hat{\phi}_L) \right).
\label{eq:seq_LP}
\end{equation}
Such matrices have $L = \frac{N_p}{2}$ columns, which, according to corollary~\ref{cor:minobs} is minimal for identifiability to be possible, and is thus the minimal duration of the pilot sequences, that corresponds to the number of estimated channel paths. In practice, the number of estimated paths is small (rarely more than ten), so that such sequences are short when compared to the ones required when using the least squares model (of length $N_t$). Note that equations~\eqref{eq:optim_LP} and \eqref{eq:seq_LP} allow to retrieve the variance and pilot sequences proposed in \cite{Gao2015} and \cite{Xie2017} (without proof of optimality).
~\newline {\bf The bias problem.} It is important to mention that the model simplification implied by the angle-constrained estimation strategy induces bias. Indeed, the obtained channel estimate $\mathbf{h}_{\text{AC}}(\hat{\boldsymbol{\theta}}_{\text{AC}})$ is constrained to belong to the range of the matrix $\hat{\mathbf{E}}$, which is not necessarily the case for the true channel $\mathbf{h}$ since the azimuth estimation is not error free and the azimuth angles may have changed since their estimation. We indeed have
$$
\big\Vert \mathbf{h} - \mathbf{h}_{\text{AC}}(\hat{\boldsymbol{\theta}}_{\text{AC}}) \big\Vert_2^2 \geq \big\Vert \mathbf{h} - \hat{\mathbf{E}}(\hat{\mathbf{E}}^H\hat{\mathbf{E}})^{-1}\hat{\mathbf{E}}^H \mathbf{h} \big\Vert_2^2,
$$
$\hat{\mathbf{E}}(\hat{\mathbf{E}}^H\hat{\mathbf{E}})^{-1}\hat{\mathbf{E}}^H$ being the matrix representing orthogonal projection onto the range of $\hat{\mathbf{E}}$. Combined with the CRB, we get the following bound on the MSE for the angle-constrained channel estimation strategy:
\begin{equation}
\text{MSE}(\hat{\boldsymbol{\theta}}_{\text{AC}}) \geq \text{max}\left(\big\Vert \mathbf{h} - \hat{\mathbf{E}}(\hat{\mathbf{E}}^H\hat{\mathbf{E}})^{-1}\hat{\mathbf{E}}^H \mathbf{h} \big\Vert_2^2, \frac{\sigma^2L^2}{P_t}\right).
\label{eq:mse_ac}
\end{equation}
A high level summary of the \newnew{channel tracking phase of the} angle-constrained channel estimation strategy for a single user in given in algorithm~\ref{alg:angle_constrained}. It fits into a multi-user framework and a complete transmission workflow nicely, as explained in \cite{Gao2015,Xie2017}, but we do not give too much details on this here since the present illustration focuses only on the downlink channel \newnew{tracking} phase.
\begin{algorithm}[htb]
\new{\caption{\new{High level summary of \newnew{channel tracking phase of the} angle-constrained estimation strategy}}
\begin{algorithmic}[1] 
\REQUIRE Estimates $\hat{\phi}_1,\dots,\hat{\phi}_L$ of the azimuth angles (obtained through previous downlink or uplink channel estimates).
\STATE Build $\hat{\mathbf{E}} = (\mathbf{e}(\hat{\phi}_1),\dots,\mathbf{e}(\hat{\phi}_L))$.
\STATE Send pilot sequences $\mathbf{X}_{\text{AC}}  =   \sqrt{\frac{P_t}{L}}\left(\mathbf{e}(\hat{\phi}_1),\dots,\mathbf{e}(\hat{\phi}_L) \right)$ of duration $L$ and receive user feedback to build observations following \eqref{eq:observation}.
\STATE Estimate path gains $\hat{\beta}_1,\dots,\hat{\beta}_L$.
\STATE Estimate channel $\mathbf{h}_{\text{AC}}(\hat{\boldsymbol{\theta}}_{\text{AC}}) = \sum_{l=1}^L\hat{\beta}_l \mathbf{e}(\hat{\phi}_l)$.
\ENSURE $\text{MSE}(\hat{\boldsymbol{\theta}}_{\text{AC}}) \geq \text{max}\left(\big\Vert \mathbf{h} - \hat{\mathbf{E}}(\hat{\mathbf{E}}^H\hat{\mathbf{E}})^{-1}\hat{\mathbf{E}}^H \mathbf{h} \big\Vert_2^2, \frac{\sigma^2L^2}{P_t}\right)$
\end{algorithmic}
\label{alg:angle_constrained}}
\end{algorithm}



\new{\subsection{A new channel estimation strategy for physical models}
\label{ssec:newstrat}
Let us now propose another strategy \newnew{for the channel tracking phase} that \newnew{is more robust to the angle estimation error because it} does not suffer from the same bias problem as the angle-constrained strategy. \newnew{Moreover, the proposed strategy} also leads to short pilot sequences \newnew{thanks to the use of theorem~\ref{thm:optim} and algorithm~\ref{alg:single} for pilot design,} and incurs only a small variance increase \newnew{compared to the angle-constrained strategy}.
We indeed argue that is is possible to make better use of previously acquired azimuth angles estimates $\hat{\phi}_1,\dots,\hat{\phi}_L$ than to bias the model by considering them perfectly estimated and to constrain the channel estimate to belong to the range of $\hat{\mathbf{E}}$. Indeed, one could instead use the azimuth estimates to design pilot sequences according to algorithm~\ref{alg:single}, so as to update the angle estimation while estimating the channel. This amounts to consider the physical model of \eqref{eq:physical_params_new} and \eqref{eq:physical_channel} as is, without fixing the angles. This way, the variation space is expressed as in \eqref{eq:variation_phy}.
\newline {\bf Optimal CRB.} Applying theorem~\ref{thm:optim}, the optimal CRB of this full physical model is thus bounded as 
\begin{equation}
\frac{9}{4}\times\frac{\sigma^2L^2}{P_t}\leq \text{CRB}_\text{min}(\boldsymbol{\theta}_{\text{PHY}}) \leq \frac{9}{2}\times \frac{\sigma^2L^2}{P_t}.
\label{eq:optim_PHY}
\end{equation}
\newnew{These bounds are obtained by injecting $N_p = 3L$ into \eqref{eq:bound_CRB}.
The exact value of the CRB depends} on the values of the scalars $c_1,\dots,c_{\lfloor \frac{N_p}{2}\rfloor}$. These bounds are $\frac{9}{4}$ times larger than the CRB bounds of the angle-constrained strategy \eqref{eq:optim_LP}, because of the larger number of real parameters to estimate ($3L$ instead of $2L$). It is attained for perfect estimates of the azimuth angles. A natural estimate of the variation space based on azimuth estimates reads
\begin{equation}
\hat{\mathcal{V}}_{\boldsymbol{\theta}_{\text{PHY}}}  = \text{span}_\mathbb{R} \left(\Big\{\mathbf{e}(\hat{\phi}_l),-\mathrm{j}\mathbf{e}(\hat{\phi}_l),\frac{\partial\mathbf{e}(\hat{\phi}_l)}{\partial \hat{\phi}_l}\Big\}_{l=1}^L\right).
\label{eq:variation_phy_est}
\end{equation}
This estimate can be used directly to design pilot sequences according to algorithm~\ref{alg:single}.
 The proposed \newnew{channel tracking} strategy is summarized in algorithm~\ref{alg:new_strat}. It fits into a multi-user framework and a complete transmission strategy exactly as its angle-constrained counterpart, to which it is pretty similar. The two main differences with algorithm~\ref{alg:angle_constrained} are the slightly longer pilot sequences ($\lceil \frac{3L}{2} \rceil$ instead of $L$) and the fact that azimuth angles are not kept fixed but updated. \newnew{The update of the angles estimates can be done for example using the maximum likelihood principle with exhaustive testing of angles sampling uniformly a small region around the previous angles estimates, or with a few steps of gradient descent initialized at the previous angles estimates.} This update has the effect of suppressing the bias and leads to a lower bound on the MSE comprising only variance, which vanishes at high SNR. \newnew{The two evoked strategies (angle-constrained and proposed) actually differ in the way prior information is used. For the angle-constrained strategy, channel estimates themselves are based on fixed and previously acquired angle estimates, while for the proposed strategy, only the downlink pilot sequences are based on previous angle estimates, which is more robust as shown empirically in the following paragraph, because pilot sequences stay quasi-optimal if the true angles stay in a neighborhood of the estimates. Moreover, the proposed strategy allows an update of the angle estimates, which may reduce the frequency at which the path estimation phase (which is computationally heavy) has to be carried out.}
\begin{algorithm}[htb]
\new{\caption{\new{High level summary of \newnew{channel tracking phase of the} proposed estimation strategy ($L$ even)}}
\begin{algorithmic}[1] 
\REQUIRE Estimates $\hat{\phi}_1,\dots,\hat{\phi}_L$ of the azimuth angles (obtained through previous downlink or uplink channel estimates).
\STATE Build $\hat{\mathcal{V}}_{\boldsymbol{\theta}_{\text{PHY}}}  = \text{span}_\mathbb{R} \left(\left\{\mathbf{e}(\hat{\phi}_l),-\mathrm{j}\mathbf{e}(\hat{\phi}_l),\frac{\partial\mathbf{e}(\hat{\phi}_l)}{\partial \hat{\phi}_l}\right\}_{l=1}^L\right)$.
\STATE Send pilot sequences $\mathbf{X}_{\text{PHY}}$ of duration $\lceil\frac{3L}{2}\rceil$ built using algorithm~\ref{alg:single} and receive user feedback to build observations following \eqref{eq:observation}.
\STATE Estimate path gains $\hat{\beta}_1,\dots,\hat{\beta}_L$ and update $\hat{\phi}_1,\dots,\hat{\phi}_L$.
\STATE Estimate channel $\mathbf{h}_{\text{PHY}}(\hat{\boldsymbol{\theta}}_{\text{PHY}}) = \sum_{l=1}^L\hat{\beta}_l \mathbf{e}(\hat{\phi}_l)$.
\ENSURE $\text{MSE}(\hat{\boldsymbol{\theta}}_{\text{PHY}}) \geq  \frac{9}{4}\times\frac{\sigma^2L^2}{P_t}$
\end{algorithmic}
\label{alg:new_strat}}
\end{algorithm}}

\noindent{\bf Comparison of strategies with azimuth error.} Let us now compare numerically the proposed strategy with the angle-constrained estimation strategy, taking into account the azimuth estimation error, \newnew{in order to assess the robustness of the strategies with respect to the angle estimation error}. To do so, we consider $N_t=64$ antennas at the base station (half-wavelength separated ULA).
First, to illustrate the fundamental difference on a simple example, we consider a single path channel $\mathbf{h} = \beta\mathbf{e}(\phi)$ with estimated azimuth $\hat{\phi}=0$ and true azimuth $\phi = \hat{\phi} + \Delta$, $\Delta$ being the azimuth estimation error. Then, a lower bound on the relative MSE (MSE divided by the squared norm of the channel) is computed for both strategies. Regarding the new proposed strategy, the lower bound is simply the relative CRB, which is the bound of theorem~\ref{thm:CRB} divided by the squared norm of the channel:
\begin{equation}
\frac{\sigma^2}{2\left\Vert \mathbf{h} \right\Vert_2^2}\text{Tr}\left[\mathfrak{Re}\Big\{\mathbf{U}^H \mathbf{M}\mathbf{M}^H \mathbf{U}\Big\}^{-1} \right].
\label{eq:relCRB}
\end{equation}
with
$$
\mathbf{U} = \left( \mathbf{e}(\phi),\overline{\frac{\partial\mathbf{e}(\phi)}{\partial \phi}}\right), \mathbf{M} = \sqrt{\frac{P_t}{\sqrt{2}+1}}\left(2^{\frac{1}{4}}\mathbf{e}(\hat{\phi}),\overline{\frac{\partial\mathbf{e}(\hat{\phi})}{\partial \hat{\phi}}} \right),
$$
where $\overline{\frac{\partial\mathbf{e}(\phi)}{\partial \phi}}$ is simply the normalized version of $\frac{\partial\mathbf{e}(\phi)}{\partial \phi}$ (this observation matrix is the result of applying algorithm~\ref{alg:single}).
In this case, the physical model comprises $N_p = 3$ parameters, leading to sequences of length $\lceil \frac{3L}{2}\rceil=2$.
Regarding the angle-constrained estimation strategy, the lower bound on the relative MSE is computed as the maximum of \eqref{eq:relCRB} and the relative bias
\begin{equation}
\frac{\big\Vert \mathbf{h} - \hat{\mathbf{E}}(\hat{\mathbf{E}}^H\hat{\mathbf{E}})^{-1}\hat{\mathbf{E}}^H \mathbf{h} \big\Vert_2^2}{\left\Vert \mathbf{h} \right\Vert_2^2},
\label{eq:relbias}
\end{equation}
with
$$
\mathbf{U} = \mathbf{e}(\phi), \hat{\mathbf{E}} = \mathbf{e}(\hat{\phi}), \mathbf{M} = \sqrt{P_t}\mathbf{e}(\hat{\phi})
$$
In this case, the simplified physical model comprises $N_p = 2$ parameters, leading to sequences of length $L=1$.
These error lower bounds are shown for \newnew{$\Delta \in \{0^\circ, 1.0^\circ,5.0^\circ\}$ ($0^\circ$ meaning perfect angle estimates, leading to the optimal CRB)} on figure~\ref{fig:onepath} as a function of the potential signal to noise ratio (pSNR)
$$
\text{pSNR} \triangleq \frac{P_t\left\Vert \mathbf{h} \right\Vert_2^2}{\sigma^2},
$$
which is an upper bound on the classical SNR, attained only if the precoder is perfectly collinear to the channel. From the figure, \newnew{one can notice that the optimal CRB (with perfect angle estimate, $\Delta=0$) is lower for the angle-constrained strategy than for the proposed strategy, as stated in the previous paragraph. However, it is interesting to notice that, as soon as there is some angle estimation error ($\Delta>0$)}, the proposed strategy is always theoretically better than the angle constrained strategy at high pSNR. This is because it is not biased toward the previously estimated azimuth angles, as is the angle-constrained strategy, since the bias is independent of the pSNR. Moreover, at low pSNR, the proposed strategy and the angle-constrained one are \newnew{within a few decibels, because the number of parameter to estimate is pretty close in both cases (proportional to the number of dominant paths). For a relatively small angle estimation error ($\Delta=1$), the proposed strategy exhibits a robust behavior, being very close to the optimal CRB ($\Delta=0$) for all pSNR values, which is not the case of the angle-constrained strategy. Finally, note that for a pretty large angle estimation error ($\Delta=5$) the proposed strategy leads to better results than the angle-constrained one, and that even for a low pSNR.} In summary, the azimuth estimation error seems to have much less effect on the performance of the proposed strategy than on that of the angle-constrained strategy. Once again, this is explained by the fact that updating the azimuths estimates allows to reduce the impact of the initial estimation error, \newnew{yielding a more robust strategy. In the two phases communication workflow we consider here, the robustness is interesting since it allows to perform the channel tracking phase longer without degrading the link. In turn, this allows to avoid as long as possible the costly path estimation phase, in order to spare resources for communication purposes.}

\begin{figure}[htbp]
\includegraphics[width=\columnwidth]{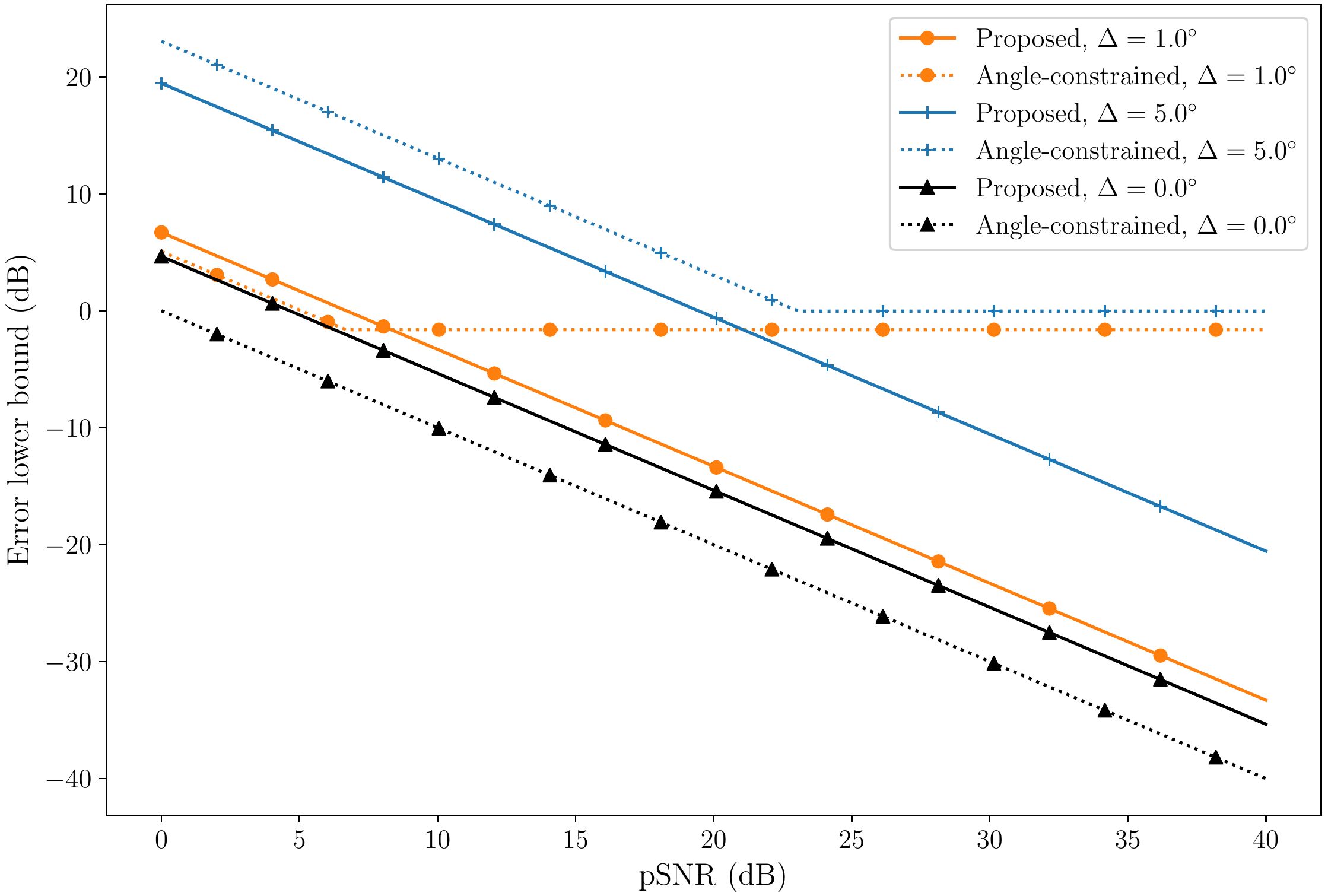}
\new{\caption{Comparison of estimation strategies for single path channels ($L=1$).}\label{fig:onepath}}
\end{figure}

\new{
In order to further validate the approach, let us now study a more practical scenario, with multipath channels. To do so, we consider a clustered channel model at a frequency of $28$~GHz, with $L$ being equal to the number of clusters. The number of clusters and their powers are drawn according to the NYUSIM channel model \cite{Samimi2016} in the NLOS scenario, which yields $L \in [1,7]$. Azimuth angles corresponding to the main azimuth of each cluster $\phi_1,\dots,\phi_L$ are uniformly distributed between $0$ and $2\pi$. In order to simulate the azimuth estimation error, the estimated azimuths are generated as $\hat{\phi}_l = \phi_l + \delta_l $, $\delta_l$ being uniformly distributed between $-\Delta$ and $\Delta$. Pilot sequences are built using algorithm~\ref{alg:single}. Results for \newnew{$\Delta \in \{0^\circ, 1.0^\circ, 5.0^\circ\}$} are shown on figure~\ref{fig:multipath}. The curves exhibit qualitatively the same behavior as for the single path case studied before, definitely showing that the \newnew{theoretical study performed in this paper allows to design efficient downlink pilot sequences yielding more robustness to the channel tracking phase of methods based on previous angle estimates \cite{Gao2015,Xie2017}.}
}
\begin{figure}[htbp]
\includegraphics[width=\columnwidth]{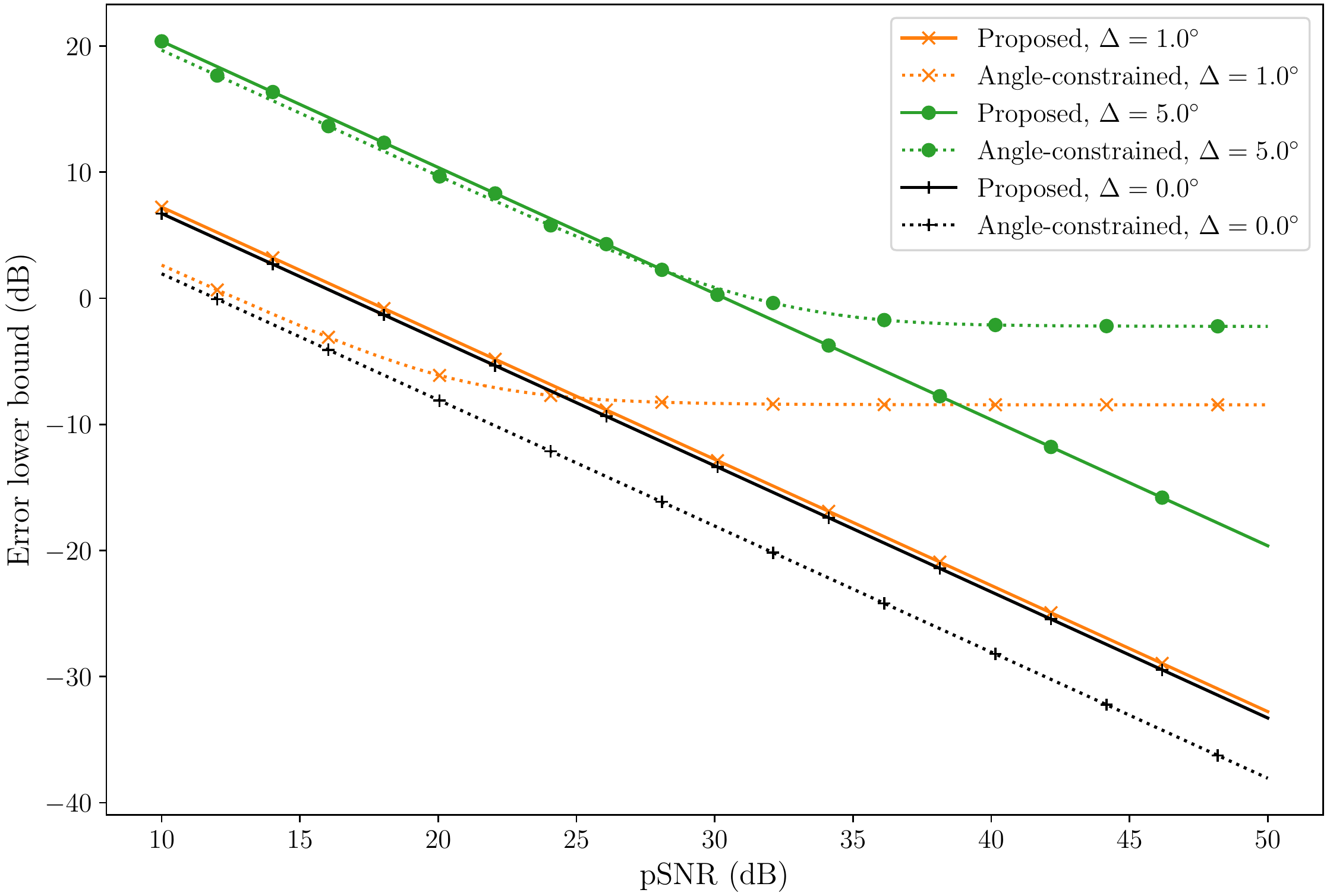}
\new{\caption{Comparison of estimation strategies for multipath channels generated according to NYUSIM  channel model \cite{Samimi2016}. Averages over $1000$ channel realizations are shown.}
\label{fig:multipath}}
\end{figure}

\section{Conclusion}
\label{sec:conclusion}
In this paper, we studied the problem of estimating a channel of interest parameterized according to a nonlinear model, based on noisy complex linear measurements, obeying \eqref{eq:observation}.

The Cram\'er-Rao bound of such a general problem is established, showing its key dependency on an $\mathbb{R}$-vector space we called \emph{variation space} (theorem~\ref{thm:CRB}). The CRB is shown to be proportional to the trace of the inverse compression of the observation matrix to the variation space (corollary~\ref{cor:intrinsic_CRB}).

\moved{The identifiability conditions on the observation matrix are given (theorem~\ref{thm:identif}), as well as a minimal number of measurements for identifiability to be possible (corollary~\ref{cor:minobs}).}

A general result about $\mathbb{R}$-vector spaces is provided (lemma~\ref{lem:decomp_V_theta}), which allows to decompose the variation space into $\mathbb{C}$-orthogonal subspaces. Such a decomposition proves useful in the study of optimal observation matrices which is carried out next.

The minimal CRB and associated observation matrices of minimal length are determined (theorem~\ref{thm:optim}). They are shown to depend only on the observation power, the noise level and intrinsic properties of the variation space.

The results obtained for the general estimation problem are then particularized to MIMO channel estimation. It is shown that the general framework allows to retrieve well-known results, but also to derive optimal pilot sequences of minimal length in a setting for which they had not been determined yet. \newnew{In particular, it  allows to design efficient downlink pilot sequences yielding more robustness to the channel tracking phase of methods based on previous angle estimates \cite{Gao2015,Xie2017}.} 

In the future, the theoretical results provided here could be applied to more practical MIMO systems, for example including hybrid precoding and combining \cite{Heath2016,Elayach2014,Sayeed2016}. They could allow to determine optimal pilot sequences in this context, as well as to quantify the suboptimality of existing or simpler schemes.
They could also very well be applied outside the MIMO channel estimation scope, for any estimation problem whose observation model fits  \eqref{eq:observation}. 
Note that the strength of this study lies in its generality, since it encompasses all deterministic models, linear or not, depending on real or complex parameters, and is valid for any unbiased estimator. This renders the obtained results potentially useful well beyond the scope of channel estimation. 

\appendices

\moved{
\section{Proof of lemma~\ref{lem:decomp_V_theta}}
\label{proof:lemma1}
(i) The general strategy of the proof is to exhibit an $\mathbb{R}$-orthonormal basis of $\mathcal{E}$ in which vectors can be grouped by two so that vectors of different groups are $\mathbb{C}$-orthogonal. The first step of the proof amounts to link the real and complex inner products as
\begin{equation}
\langle \mathbf{a},\mathbf{b} \rangle_{\mathbb{C}}= \mathfrak{Re}\{\mathbf{a}^H\mathbf{b}\} + \mathrm{j}\mathfrak{Im}\{\mathbf{a}^H\mathbf{b}\} = \langle \mathbf{a},\mathbf{b} \rangle_{\mathbb{R}} + \mathrm{j}\langle \mathrm{j}\mathbf{a},\mathbf{b} \rangle_{\mathbb{R}}.
\label{eq:decomp_scalar}
\end{equation}
Now, the idea is to maximize the second term of this sum (which will automatically cancel the first one) in order to build recursively an $\mathbb{R}$-orthonormal basis of $\mathcal{E}$ with the sought properties. To do so, let us choose
\begin{equation}
(\mathbf{v}_{1},\mathbf{w}_{1}) \in \underset{\underset{\left\Vert \mathbf{v} \right\Vert_2=\left\Vert \mathbf{w} \right\Vert_2=1}{(\mathbf{v},\mathbf{w})\in \mathcal{E}^2}}{\text{argmax}} \langle \mathbf{v},\mathrm{j}\mathbf{w} \rangle_{\mathbb{R}},
\label{eq:max_scalar}
\end{equation}
which necessarily exist since the function $\langle \mathbf{v},\mathrm{j}\mathbf{w} \rangle_{\mathbb{R}}$ to maximize is continuous and the constraint set is compact. Moreover, let
\begin{equation}
c_1\triangleq \underset{\underset{\left\Vert \mathbf{v} \right\Vert_2=\left\Vert \mathbf{w} \right\Vert_2=1}{(\mathbf{v},\mathbf{w})\in \mathcal{E}^2}}{\text{max}} \langle \mathbf{v},\mathrm{j}\mathbf{w} \rangle_{\mathbb{R}} = \langle \mathbf{v}_{1},\mathrm{j}\mathbf{w}_{1} \rangle_{\mathbb{R}}.
\label{eq:c_max}
\end{equation}
Note that by the Cauchy-Schwarz inequality, $c_1 \leq 1$. Moreover, $c_1 \geq 0$ because if $\langle \mathbf{v},\mathrm{j}\mathbf{w} \rangle_{\mathbb{R}} \leq 0$, then $\langle -\mathbf{v},\mathrm{j}\mathbf{w} \rangle_{\mathbb{R}} \geq \langle \mathbf{v},\mathrm{j}\mathbf{w} \rangle_{\mathbb{R}}$. The case $c_1 = 0$ is easily handled since in that case, any $\mathbb{R}$-orthogonal basis is automatically also $\mathbb{C}$-orthogonal and (i) is proven. For the case $0<c_1 \leq 1$, let us write the Lagrangian of the constrained maximization problem \eqref{eq:max_scalar}:
\begin{equation*}
\mathcal{L}(\mathbf{v},\mathbf{w},\alpha,\beta) \triangleq 
\langle \mathbf{v},\mathrm{j}\mathbf{w} \rangle_{\mathbb{R}} + \alpha(\langle \mathbf{v},\mathbf{v} \rangle_{\mathbb{R}}-1) + \beta(\langle \mathbf{w},\mathbf{w} \rangle_{\mathbb{R}}-1),
\end{equation*}
where $\alpha \in \mathbb{R}$ and $\beta \in \mathbb{R}$ are the Lagrange multipliers (voluntarily ignoring the constraint $(\mathbf{v},\mathbf{w})\in \mathcal{E}^2$ for now). 
Differentiating it with respect to $\mathbf{v}$ and $\mathbf{w}$ and writing the optimality conditions (introducing the constraint $(\mathbf{v},\mathbf{w})\in \mathcal{E}^2$) yields $\forall \mathbf{z} \in \mathcal{E}$,
\begin{equation}
\langle \mathrm{j}\mathbf{w}_{1}+ 2\alpha \mathbf{v}_{1} , \mathbf{z}\rangle_{\mathbb{R}} = 0
\label{eq:first_cond}
\end{equation}
and
\begin{equation}
\langle \mathrm{j}\mathbf{v}_{1}+ 2\beta \mathbf{w}_{1} , \mathbf{z}\rangle_{\mathbb{R}} = 0
\label{eq:sec_cond}
\end{equation}
From there, injecting $\mathbf{z} = \mathbf{v}_1$ in \eqref{eq:first_cond} (resp.\ $\mathbf{z} = \mathbf{w}_1$ in \eqref{eq:sec_cond}) yields $-2\alpha = c_1$ (resp.\ $2\beta = c_1$). Moreover, injecting $\mathbf{z} = \mathbf{w}_1$ in \eqref{eq:first_cond} yields $\mathbf{v}_{1} \perp_{\mathbb{R}} \mathbf{w}_{1}$, so that $\text{span}_{\mathbb{R}}(\{\mathbf{v}_{1},\mathbf{w}_{1}\})$ is of dimension two. 
Moreover, if $\mathbf{z} \in \mathcal{E}$ is $\mathbb{R}$-orthogonal to both $\mathbf{v}_{1}$ and $\mathbf{w}_{1}$, then \eqref{eq:first_cond} implies that $\mathbf{z} \perp_{\mathbb{C}} \mathbf{w}_{1}$ and \eqref{eq:sec_cond} implies that $\mathbf{z} \perp_{\mathbb{C}} \mathbf{v}_{1}$. 
This means that $\mathcal{E}$ can be decomposed into the direct sum of a subspace of dimension $2$ ($\text{span}_{\mathbb{R}}(\{\mathbf{v}_{1},\mathbf{w}_{1}\})$) and a subspace of dimension $d-2$ (containing all the $\mathbf{z} \in \mathcal{E}$ that are $\mathbb{R}$-orthogonal to both $\mathbf{v}_{1}$ and $\mathbf{w}_{1}$) that are $\mathbb{C}$-orthogonal. The exact same reasoning can then be re-applied to the subspace of dimension $d-2$ to prove the lemma by descent, introducing the vectors $\mathbf{v}_2$ and $\mathbf{w}_2$ as the solution of \eqref{eq:max_scalar} on this subspace and the quantity $c_2$ as the inner product $\langle \mathbf{v}_2, \mathrm{j}\mathbf{w}_2 \rangle$. The descent stops when the dimension of the remaining subspace is strictly smaller than two, so that if $d$ is odd, the last subspace of the decomposition is of dimension one.}

\moved{(ii) Now, let us prove that the subspace of dimension two identified at each step necessarily belongs to an eigenspace of the operator  $\mathbf{P}_{\mathcal{E}}\circ\mathbf{P}_{\mathrm{j}\mathcal{E}}$. First of all, by the Hilbert projection theorem, for any $\mathbf{x}\in \mathcal{F}$ we can define $\mathbf{P}_{\mathcal{E}}\mathbf{x} \triangleq \text{argmin}_{\mathbf{s}\in \mathcal{E}} \left\Vert \mathbf{x} - \mathbf{s} \right\Vert_2$ and $\mathbf{P}_{\mathrm{j}\mathcal{E}}\mathbf{x} \triangleq \text{argmin}_{\mathbf{s}\in \mathrm{j}\mathcal{E}} \left\Vert \mathbf{x} - \mathbf{s} \right\Vert_2$, which are the orthogonal projections of $\mathbf{x}$ onto $\mathcal{E}$ and $\mathrm{j}\mathcal{E}$. One can notice that the two projections are linked since $\mathbf{P}_{\mathcal{E}}(\mathrm{j}\mathbf{x}) = \mathrm{j}\mathbf{P}_{\mathrm{j}\mathcal{E}}\mathbf{x}$. Then, combining the definition of the projection operators with \eqref{eq:max_scalar} and \eqref{eq:c_max} yields
$$
c_{1}\mathbf{v}_{1} =  \mathbf{P}_{\mathcal{E}} (\mathrm{j}\mathbf{w}_{1}) = \mathrm{j} \mathbf{P}_{\mathrm{j}\mathcal{E}} \mathbf{w}_{1}
$$
and
$$
c_{1}\mathbf{w}_{1} =  \mathbf{P}_{\mathcal{E}} (\mathrm{j}\mathbf{v}_{1}) = \mathrm{j} \mathbf{P}_{\mathrm{j}\mathcal{E}} \mathbf{v}_{1}.
$$
Combining these two equations, we get
$$
\mathbf{P}_{\mathcal{E}}\circ\mathbf{P}_{\mathrm{j}\mathcal{E}} (\mathbf{v}_{1}) = -c_{1}^2\mathbf{v}_{1}
$$
and
$$
\mathbf{P}_{\mathcal{E}}\circ\mathbf{P}_{\mathrm{j}\mathcal{E}} (\mathbf{w}_{1}) = -c_{1}^2\mathbf{w}_{1},
$$
which proves our claim for the first step of the descent. The exact same reasoning can be applied at each subsequent step of the descent.}

\moved{
It is interesting to notice that another (more algebraic) proof of this lemma is possible, which gives a practical way to obtain the basis vectors corresponding to the decomposition. Indeed, let $\mathbf{U}$ be any matrix whose columns form an $\mathbb{R}$-orthonormal basis of $\mathcal{E}$. Then,
$$
\mathbf{U}^H\mathbf{U} = \mathbf{Id} + \mathrm{j}\mathbf{A},
$$
where the matrix $\mathbf{A} = \mathfrak{Im}\{\mathbf{U}^H\mathbf{U}\}$ is skew-symmetric, so that it admits the following real normal form \cite[Theorem 8.16]{Greub1975} also known as the Youla decomposition \cite{Youla1961} :
\begin{equation}
\mathbf{B}^T\mathbf{A}\mathbf{B} = 
\begin{pmatrix}
0&-c_1 & \\
c_1 & 0 & \\
& & 0&-c_2 \\
& & c_2& 0 \\ 
& & & &\ddots \\
\end{pmatrix} \triangleq \boldsymbol{\Gamma},
\label{eq:youla_proof} 
\end{equation}
with $\mathbf{B} \in\mathbb{R}^{d\times d}$ a real orthogonal matrix ($\mathbf{B}^T\mathbf{B} = \mathbf{Id}$) whose columns are eigenvectors of the symmetric positive semi-definite matrix $\mathbf{A}^T\mathbf{A} = -\mathbf{A}^2$ (whose nonzero eigenvalues are all of multiplicity two and correspond to $c_1^2,c_2^2,\dots$), and where $0\leq c_k \leq 1$, $\forall k$. This yields
$$
\mathbf{B}^T\mathbf{U}^H\mathbf{U}\mathbf{B} = \mathbf{Id}  + \mathrm{j}\boldsymbol{\Gamma} =
\begin{pmatrix}
1&-\mathrm{j}c_1 & \\
\mathrm{j}c_1 & 1 & \\
& & 1&-\mathrm{j}c_2 \\
& & \mathrm{j}c_2& 1 \\ 
& & & &\ddots \\
\end{pmatrix},
$$
so that the columns of the matrix $\mathbf{U}\mathbf{B}$ form an $\mathbb{R}$-orthonormal basis of $\mathcal{E}$ in which vectors can be grouped by two so that vectors of different groups are $\mathbb{C}$-orthogonal. This is exactly the main claim of lemma~\ref{lem:decomp_V_theta}. In practice, the matrix $\mathbf{B}$ and the values $c_1,c_2,\dots$ can be obtained by computing the real Schur decomposition of the matrix $\mathfrak{Im}\{\mathbf{U}^H\mathbf{U}\}$ and reordering the blocks.}

\moved{We preferred giving a geometric proof here in order to give more insight on the interaction between $\mathbb{R}$-vector spaces and $\mathbb{C}$-vector spaces.
Indeed, our proof highlights the fact that the quantity $c_{i}$ can be nicely interpreted as the squared cosine of the $i$-th principal angle \cite{Bjorck1973} between $\mathcal{E}$ and $\mathrm{j}\mathcal{E}$. 
}

\moved{
\section{Proof of theorem~\ref{thm:optim}}
\label{proof:thm3}
 Let us first consider the case where $N_p$ is even. Starting from the result of theorem~\ref{thm:CRB} and using the fact that it holds true for any matrix whose columns form an $\mathbb{R}$-orthonormal basis of $\mathcal{V}_{\boldsymbol{\theta}}$, we express the CRB as
$$
\text{CRB}(\boldsymbol{\theta},\mathbf{M}) = \frac{\sigma^2}{2}\text{Tr}\left[\mathfrak{Re}\Big\{\mathbf{V}^H \mathbf{M}\mathbf{M}^H \mathbf{V}\Big\}^{-1} \right],
$$
where $\mathbf{V}$ is the matrix defined in~\eqref{eq:matrix_V} when applying lemma~\ref{lem:decomp_V_theta} to $\mathcal{V}_{\boldsymbol{\theta}}$.}

\moved{Next, using the fact that for a symmetric positive semidefinite matrix $\mathbf{A}$, $(\mathbf{A}^{-1})_{ii} \geq \frac{1}{a_{ii}}, \forall i$ \cite[Theorem 7.7.15]{Horn2012}, we get
$$
\text{Tr}\left[\mathfrak{Re}\Big\{\mathbf{V}^H \mathbf{M}\mathbf{M}^H \mathbf{V}\Big\}^{-1} \right] \geq \sum_{k=1}^{\frac{N_p}{2}}\frac{1}{\left\Vert \mathbf{M}^H\mathbf{v}_k \right\Vert_2^2}+\frac{1}{\left\Vert \mathbf{M}^H\mathbf{w}_k \right\Vert_2^2},
$$
with an equality if and only if the matrix $\mathfrak{Re}\Big\{\mathbf{V}^H \mathbf{M}\mathbf{M}^H \mathbf{V}\Big\}$ is diagonal.
In order to proceed, let us define
$$
\tilde{\mathbf{u}}_k^+ = \frac{1}{\sqrt{2(1+c_k)}}(\mathbf{v}_k + \mathrm{j}\mathbf{w}_k)
$$
and
$$
\tilde{\mathbf{u}}_k^- = \frac{1}{\sqrt{2(1-c_k)}}(\mathbf{v}_k - \mathrm{j}\mathbf{w}_k),
$$
which are unitary vectors such that $\tilde{\mathbf{u}}_k^+ \perp_{\mathbb{C}} \tilde{\mathbf{u}}_k^-$, $\forall k$. These vectors allow to express
\begin{align*}
\left\Vert \mathbf{M}^H\mathbf{v}_k \right\Vert_2^2 = \frac{1}{2}\Big[&(1+c_k)\left\Vert\mathbf{M}^H\tilde{\mathbf{u}}_k^+ \right\Vert_2^2 + (1-c_k)\left\Vert\mathbf{M}^H\tilde{\mathbf{u}}_k^-\right\Vert_2^2\\
& + \sqrt{1-c_k^2}\mathfrak{Re}\{(\tilde{\mathbf{u}}_k^+)^H\mathbf{MM}^H\tilde{\mathbf{u}}_k^-\} \Big],
\end{align*}
and
\begin{align*}
\left\Vert \mathbf{M}^H\mathbf{w}_k \right\Vert_2^2 = \frac{1}{2}\Big[&(1+c_k)\left\Vert\mathbf{M}^H\tilde{\mathbf{u}}_k^+ \right\Vert_2^2 + (1-c_k)\left\Vert\mathbf{M}^H\tilde{\mathbf{u}}_k^-\right\Vert_2^2\\
& - \sqrt{1-c_k^2}\mathfrak{Re}\{(\tilde{\mathbf{u}}_k^+)^H\mathbf{MM}^H\tilde{\mathbf{u}}_k^-\} \Big].
\end{align*}
Now, let us define $P_k^+ \triangleq \left\Vert\mathbf{M}^H\tilde{\mathbf{u}}_k^+ \right\Vert_2^2$, $P_k^- \triangleq \left\Vert\mathbf{M}^H\tilde{\mathbf{u}}_k^- \right\Vert_2^2$ and $d_k \triangleq \sqrt{1-c_k^2}\mathfrak{Re}\{(\tilde{\mathbf{u}}_k^+)^H\mathbf{MM}^H\tilde{\mathbf{u}}_k^-\}$, so that we have
\begin{align*}
&\sum\nolimits_{k=1}^{\frac{N_p}{2}}\frac{1}{\left\Vert \mathbf{M}^H\mathbf{v}_k \right\Vert_2^2}+\frac{1}{\left\Vert \mathbf{M}^H\mathbf{w}_k \right\Vert_2^2}\\
 = & \sum\nolimits_{k=1}^{\frac{N_p}{2}} \frac{2}{(1+c_k)P_k^+ + (1-c_k)P_k^- + d_k }\\
& \,\,\,\,+ \frac{2}{(1-c_k)P_k^- + (1+c_k)P_k^+ - d_k} \\
\geq& \sum\nolimits_{k=1}^{\frac{N_p}{2}} \frac{4}{(1-c_k)P_k^- + (1+c_k)P_k^+},
\end{align*}
the last inequality being a direct consequence of the fact that $\frac{1}{a+b}+\frac{1}{a-b} \geq \frac{2}{a}$ (because of the convexity of the inverse function on $\mathbb{R}_+$). It becomes an equality if and only if $d_k = \sqrt{1-c_k^2}\mathfrak{Re}\{(\tilde{\mathbf{u}}_k^+)^H\mathbf{MM}^H\tilde{\mathbf{u}}_k^-\}=0$, $\forall k$.}

\moved{In summary, we have 
$$
\text{CRB}(\boldsymbol{\theta},\mathbf{M}) = \frac{\sigma^2}{2}\sum\nolimits_{k=1}^{\frac{N_p}{2}} \frac{4}{(1-c_k)P_k^- + (1+c_k)P_k^+}
$$
if and only if $\mathfrak{Re}\Big\{\mathbf{V}^H \mathbf{M}\mathbf{M}^H \mathbf{V}\Big\}$ is diagonal and $\mathfrak{Re}\{(\tilde{\mathbf{u}}_k^+)^H\mathbf{MM}^H\tilde{\mathbf{u}}_k^-\}=0$, $\forall k$. Moreover, $\left\Vert \mathbf{M} \right\Vert_F^2 = \text{Tr}[\mathbf{MM}^H] \geq   \sum_{k=1}^{\frac{N_p}{2}} P_k^+ + P_k^-$, with an equality if and only if $\text{im}_{\mathbb{C}}(\mathbf{M}) \subset \text{span}_{\mathbb{C}}(\{\tilde{\mathbf{u}}_k^+,\tilde{\mathbf{u}}_k^-\}_{k=1}^{\frac{N_p}{2}})$.
The optimization problem~\eqref{eq:pb_optim} is thus lower-bounded by the simpler problem
\begin{align}
\begin{split}
\underset{P_k^+,P_k^-, k=1,\dots,\frac{N_p}{2}}{\text{minimize}} &\quad \sum\nolimits_{k=1}^{\frac{N_p}{2}} \frac{4}{(1-c_k)P_k^- + (1+c_k)P_k^+}, \\
\text{subject to} &\quad \sum\nolimits_{k=1}^{\frac{N_p}{2}} P_k^+ + P_k^- =P.
\end{split}
\label{eq:pb_optim_reform}
\end{align}
Let us solve this problem and then identify matrices $\mathbf{M}$ for which the optimal values of \eqref{eq:pb_optim_reform} and \eqref{eq:pb_optim} coincide. 
It is obvious that at the optimum of \eqref{eq:pb_optim_reform}, $P_k^- = 0$, $\forall k$, so that it is equivalent to solve the even simpler problem
\begin{align}
\begin{split}
\underset{P_k^+, k=1,\dots,\frac{N_p}{2}}{\text{minimize}} &\quad \sum\nolimits_{k=1}^{\frac{N_p}{2}} \frac{4}{ (1+c_k)P_k^+}, \\
\text{subject to} &\quad \sum\nolimits_{k=1}^{\frac{N_p}{2}} P_k^+=P.
\end{split}
\label{eq:pb_optim_reform_2}
\end{align}
Using the Lagrange multipliers method, it is straightforward to obtain the optimal powers
$$
(P_k^+)_\text{opt} = \frac{P}{\sqrt{1+c_k}\sum_{j=1}^{\frac{N_p}{2}}\frac{1}{\sqrt{1+c_j}}},
$$
and the optimal value of the optimization problems \eqref{eq:pb_optim_reform} and \eqref{eq:pb_optim_reform_2} is 
$$
\sum\nolimits_{k=1}^{\frac{N_p}{2}} \frac{4}{ (1+c_k)(P_k^+)_{\text{opt}}} = \frac{4}{P}\left( \sum\nolimits_{k=1}^{\frac{N_p}{2}} \frac{1}{\sqrt{1+c_k}} \right)^2.
$$
It is also the optimal value of problem~\eqref{eq:pb_optim}, since it is attained with the observation matrix
$$
\mathbf{M}_\text{opt} = \left(\sqrt{\left(P_1^+\right)_{\text{opt}}} \tilde{\mathbf{u}}_1^+,\dots, \sqrt{\left(P_{\frac{N_p}{2}}^+\right)_{\text{opt}}} \tilde{\mathbf{u}}_\frac{N_p}{2}^+ \right),
$$
which indeed guarantees that $P_k^+ = (P_k^+)_\text{opt}$ and $d_k=0$, $\forall k$, that $\mathfrak{Re}\Big\{\mathbf{V}^H \mathbf{M}_\text{opt}\mathbf{M}_\text{opt}^H \mathbf{V}\Big\}$ is diagonal, and $\left\Vert \mathbf{M}_{\text{opt}} \right\Vert_F^2 = P$.}

\moved{The proof is very similar in the case where $N_p$ is odd, the only difference being that the decomposition of $\mathcal{V}_{\boldsymbol{\theta}}$ is the one given in \eqref{eq:decomp_V_theta_odd} rather than the one given in \eqref{eq:decomp_V_theta}.}

%
%
%
%
%

\ifCLASSOPTIONcaptionsoff
  \newpage
\fi



\bibliographystyle{unsrt}
\bibliography{biblio_mimo}

\begin{thebibliography}{10}

\bibitem{Haykin2008}
Simon Haykin.
\newblock {\em Communication systems}.
\newblock John Wiley \& Sons, 2008.

\bibitem{Rusek2013}
Fredrik Rusek, Daniel Persson, Buon~Kiong Lau, Erik~G Larsson, Thomas~L
  Marzetta, Ove Edfors, and Fredrik Tufvesson.
\newblock Scaling up mimo: Opportunities and challenges with very large arrays.
\newblock {\em IEEE Signal Processing Magazine}, 30(1):40--60, 2013.

\bibitem{Larsson2014}
Erik~G Larsson, Ove Edfors, Fredrik Tufvesson, and Thomas~L Marzetta.
\newblock Massive mimo for next generation wireless systems.
\newblock {\em IEEE Communications Magazine}, 52(2):186--195, 2014.

\bibitem{Lu2014}
Lu~Lu, Geoffrey~Ye Li, A~Lee Swindlehurst, Alexei Ashikhmin, and Rui Zhang.
\newblock An overview of massive mimo: Benefits and challenges.
\newblock {\em IEEE journal of selected topics in signal processing},
  8(5):742--758, 2014.

\bibitem{Rappaport2013}
Theodore~S Rappaport, Shu Sun, Rimma Mayzus, Hang Zhao, Yaniv Azar, Kevin Wang,
  George~N Wong, Jocelyn~K Schulz, Mathew Samimi, and Felix Gutierrez.
\newblock Millimeter wave mobile communications for 5g cellular: It will work!
\newblock {\em IEEE access}, 1:335--349, 2013.

\bibitem{Swindlehurst2014}
A~Lee Swindlehurst, Ender Ayanoglu, Payam Heydari, and Filippo Capolino.
\newblock Millimeter-wave massive mimo: the next wireless revolution?
\newblock {\em IEEE Communications Magazine}, 52(9):56--62, 2014.

\bibitem{Heath2016}
Robert~W Heath, Nuria Gonzalez-Prelcic, Sundeep Rangan, Wonil Roh, and Akbar~M
  Sayeed.
\newblock An overview of signal processing techniques for millimeter wave mimo
  systems.
\newblock {\em IEEE journal of selected topics in signal processing},
  10(3):436--453, 2016.

\bibitem{Biguesh2006}
Mehrzad Biguesh and Alex~B Gershman.
\newblock Training-based mimo channel estimation: a study of estimator
  tradeoffs and optimal training signals.
\newblock {\em IEEE transactions on signal processing}, 54(3):884--893, 2006.

\bibitem{Adhikary2013}
A.~{Adhikary}, J.~{Nam}, J.~{Ahn}, and G.~{Caire}.
\newblock Joint spatial division and multiplexing—the large-scale array
  regime.
\newblock {\em IEEE Transactions on Information Theory}, 59(10):6441--6463,
  2013.

\bibitem{Adhikary2014}
A.~{Adhikary}, E.~{Al Safadi}, M.~K. {Samimi}, R.~{Wang}, G.~{Caire}, T.~S.
  {Rappaport}, and A.~F. {Molisch}.
\newblock Joint spatial division and multiplexing for mm-wave channels.
\newblock {\em IEEE Journal on Selected Areas in Communications},
  32(6):1239--1255, 2014.

\bibitem{Gao2015}
Z.~{Gao}, L.~{Dai}, Z.~{Wang}, and S.~{Chen}.
\newblock Spatially common sparsity based adaptive channel estimation and
  feedback for fdd massive {MIMO}.
\newblock {\em IEEE Transactions on Signal Processing}, 63(23):6169--6183,
  2015.

\bibitem{Xie2017}
Hongxiang Xie, Feifei Gao, Shun Zhang, and Shi Jin.
\newblock A unified transmission strategy for tdd/fdd massive mimo systems with
  spatial basis expansion model.
\newblock {\em IEEE Transactions on Vehicular Technology}, 66(4):3170--3184,
  2017.

\bibitem{Rao1945}
Calyampudi~Radakrishna Rao.
\newblock Information and the accuracy attainable in the estimation of
  statistical parameters.
\newblock {\em Bulletin of the Calcutta Mathematical Society}, 37:81--89, 1945.

\bibitem{Cramer1946}
Harald Cram{\'e}r.
\newblock {\em Mathematical Methods of Statistics}, volume~9.
\newblock Princeton university press, 1946.

\bibitem{Marzetta1999}
Thomas~L Marzetta.
\newblock Blast training: Estimating channel characteristics for high capacity
  space-time wireless.
\newblock In {\em Proceedings of the Annual Allerton Conference on
  Communication Control and Computing}, volume~37, pages 958--966. Citeseer,
  1999.

\bibitem{Ma2003}
Xiaoli Ma, Georgios~B Giannakis, and Shuichi Ohno.
\newblock Optimal training for block transmissions over doubly selective
  wireless fading channels.
\newblock {\em IEEE Transactions on Signal Processing}, 51(5):1351--1366, 2003.

\bibitem{Barhumi2003}
Imad Barhumi, Geert Leus, and Marc Moonen.
\newblock Optimal training design for mimo ofdm systems in mobile wireless
  channels.
\newblock {\em IEEE Transactions on signal processing}, 51(6):1615--1624, 2003.

\bibitem{Minn2006}
Hlaing {Minn} and N.~{Al-Dhahir}.
\newblock Optimal training signals for {MIMO} {OFDM} channel estimation.
\newblock {\em IEEE Transactions on Wireless Communications}, 5(5):1158--1168,
  2006.

\bibitem{Kotecha2004b}
Jayesh~H Kotecha and Akbar~M Sayeed.
\newblock Transmit signal design for optimal estimation of correlated mimo
  channels.
\newblock {\em IEEE Transactions on Signal Processing}, 52(2):546--557, 2004.

\bibitem{Bjornson2009}
Emil Bjornson and Bj{\"o}rn Ottersten.
\newblock A framework for training-based estimation in arbitrarily correlated
  rician mimo channels with rician disturbance.
\newblock {\em IEEE Transactions on Signal Processing}, 58(3):1807--1820, 2009.

\bibitem{Choi2014}
Junil Choi, David~J Love, and Patrick Bidigare.
\newblock Downlink training techniques for fdd massive mimo systems: Open-loop
  and closed-loop training with memory.
\newblock {\em IEEE Journal of Selected Topics in Signal Processing},
  8(5):802--814, 2014.

\bibitem{Gu2019}
Y.~{Gu} and Y.~D. {Zhang}.
\newblock Information-theoretic pilot design for downlink channel estimation in
  fdd massive {MIMO} systems.
\newblock {\em IEEE Transactions on Signal Processing}, 67(9):2334--2346, 2019.

\bibitem{Bazzi2017}
Samer Bazzi and Wen Xu.
\newblock Downlink training sequence design for fdd multiuser massive mimo
  systems.
\newblock {\em IEEE Transactions on Signal Processing}, 65(18):4732--4744,
  2017.

\bibitem{Bazzi2018}
Samer Bazzi and Wen Xu.
\newblock On the amount of downlink training in correlated massive mimo
  channels.
\newblock {\em IEEE Transactions on Signal Processing}, 66(9):2286--2299, 2018.

\bibitem{Bazzi2019}
Samer Bazzi, Stelios Stefanatos, Luc Le~Magoarou, Salah~Eddine Hajri, Mohamad
  Assaad, St{\'e}phane Paquelet, Gerhard Wunder, and Wen Xu.
\newblock Exploiting the massive mimo channel structural properties for
  minimization of channel estimation error and training overhead.
\newblock {\em IEEE Access}, 7:32434--32452, 2019.

\bibitem{Lemagoarou2018}
Luc {Le Magoarou} and St{\'{e}}phane Paquelet.
\newblock Parametric channel estimation for massive {MIMO}.
\newblock In {\em IEEE Statistical Signal Processing Workshop (SSP)}, 2018.

\bibitem{Lemagoarou2019}
Luc {Le Magoarou} and St{\'e}phane {Paquelet}.
\newblock {Performance of MIMO channel estimation with a physical model}.
\newblock {\em arXiv e-prints}, page arXiv:1902.07031, Feb 2019.

\bibitem{Kay1993}
Steven~M. Kay.
\newblock {\em Fundamentals of Statistical Signal Processing: Estimation
  Theory}.
\newblock Prentice-Hall, Inc., Upper Saddle River, NJ, USA, 1993.

\bibitem{Vandenbos1994}
Adriaan Van~den Bos.
\newblock A cram{\'e}r-rao lower bound for complex parameters.
\newblock {\em IEEE Transactions on Signal Processing [see also Acoustics,
  Speech, and Signal Processing, IEEE Transactions on], 42 (10)}, 1994.

\bibitem{Slepian1954}
David Slepian.
\newblock Estimation of signal parameters in the presence of noise.
\newblock {\em Transactions of the IRE Professional Group on Information
  Theory}, 3(3):68--89, 1954.

\bibitem{Bangs1971}
G.~W. Bangs.
\newblock {\em Array Processing With Generalized Beamformers}.
\newblock PhD thesis, Yale university, CT, USA, 1971.

\bibitem{Besson2013}
Olivier Besson and Yuri~I Abramovich.
\newblock On the fisher information matrix for multivariate elliptically
  contoured distributions.
\newblock {\em IEEE Signal Processing Letters}, 20(11):1130--1133, 2013.

\bibitem{Darsena2004}
Donatella Darsena, Giacinto Gelli, Luigi Paura, and Francesco Verde.
\newblock Subspace-based blind channel identification of siso-fir systems with
  improper random inputs.
\newblock {\em Signal Processing}, 84(11):2021--2039, 2004.

\bibitem{Darsena2005}
Donatella Darsena, Giacinto Gelli, Luigi Paura, and Francesco Verde.
\newblock Widely linear equalization and blind channel identification for
  interference-contaminated multicarrier systems.
\newblock {\em IEEE Transactions on Signal Processing}, 53(3):1163--1177, 2005.

\bibitem{Abdallah2011}
Saeed Abdallah and Ioannis~N Psaromiligkos.
\newblock Widely linear versus conventional subspace-based estimation of simo
  flat-fading channels: Mean squared error analysis.
\newblock {\em IEEE Transactions on Signal Processing}, 60(3):1307--1318, 2011.

\bibitem{Delmas2004}
J-P Delmas and Habti Abeida.
\newblock Stochastic crame/spl acute/r-rao bound for noncircular signals with
  application to doa estimation.
\newblock {\em IEEE Transactions on Signal Processing}, 52(11):3192--3199,
  2004.

\bibitem{Halmos1982}
Paul~Richard Halmos.
\newblock {\em A Hilbert space problem book}, volume~19.
\newblock Springer Science \& Business Media, 1982.

\bibitem{Sayeed2002}
Akbar~M Sayeed.
\newblock Deconstructing multiantenna fading channels.
\newblock {\em IEEE Transactions on Signal Processing}, 50(10):2563--2579,
  2002.

\bibitem{Bajwa2010}
W.~U. Bajwa, J.~Haupt, A.~M. Sayeed, and R.~Nowak.
\newblock Compressed channel sensing: A new approach to estimating sparse
  multipath channels.
\newblock {\em Proceedings of the IEEE}, 98(6):1058--1076, June 2010.

\bibitem{Samimi2016}
Mathew~K Samimi and Theodore~S Rappaport.
\newblock 3-d millimeter-wave statistical channel model for 5g wireless system
  design.
\newblock {\em IEEE Transactions on Microwave Theory and Techniques},
  64(7):2207--2225, 2016.

\bibitem{Elayach2014}
Omar El~Ayach, Sridhar Rajagopal, Shadi Abu-Surra, Zhouyue Pi, and Robert~W
  Heath.
\newblock Spatially sparse precoding in millimeter wave mimo systems.
\newblock {\em IEEE Transactions on Wireless Communications}, 13(3):1499--1513,
  2014.

\bibitem{Sayeed2016}
Akbar~M. Sayeed and John~H. Brady.
\newblock {\em Millimeter-Wave MIMO Transceivers: Theory, Design and
  Implementation}, pages 231--253.
\newblock John Wiley \& Sons, Ltd, 2016.

\bibitem{Greub1975}
Werner~H Greub.
\newblock {\em Linear algebra}, volume~23.
\newblock Springer Science \& Business Media, 1975.

\bibitem{Youla1961}
D.~C. Youla.
\newblock A normal form for a matrix under the unitary congruence group.
\newblock {\em Canadian Journal of Mathematics}, 13:694–704, 1961.

\bibitem{Bjorck1973}
Åke Björck and Gene~H. Golub.
\newblock Numerical methods for computing angles between linear subspaces.
\newblock {\em Mathematics of Computation}, 27(123):579--594, 1973.

\bibitem{Horn2012}
Roger~A Horn and Charles~R Johnson.
\newblock {\em Matrix Analysis: Second Edition}.
\newblock Cambridge university press, 2012.

\end{thebibliography}

%

\begin{IEEEbiography}[{\includegraphics[width=1in,height=1.25in,viewport=70 50 420 487.5,clip,keepaspectratio]{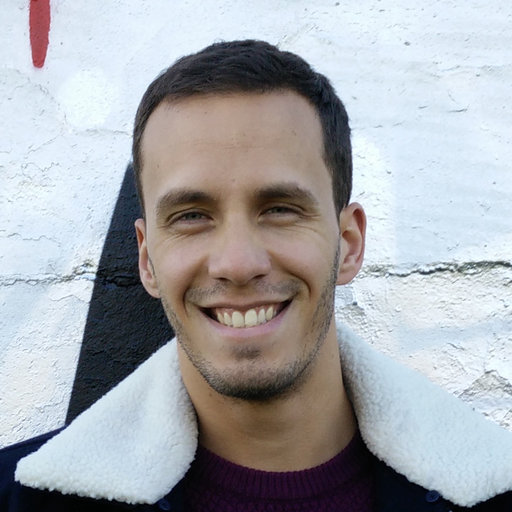}}]{Luc Le Magoarou} is a researcher at b\raisebox{0.2mm}{\scalebox{0.7}{\textbf{$<>$}}}com (Rennes, France). He received the Ph.D. in signal processing and the M.Sc. in electrical engineering, both from the National Institute of Applied Sciences (INSA Rennes, France), in 2016 and 2013 respectively. During his Ph.D., he was with Inria (Rennes, France), in the PANAMA research group. His main research interests are signal processing and machine learning, currently applied to MIMO communication systems. 
\end{IEEEbiography}
\begin{IEEEbiography}[{\includegraphics[width=1in,height=1.25in,viewport=50 50 540 662.5,clip,keepaspectratio]{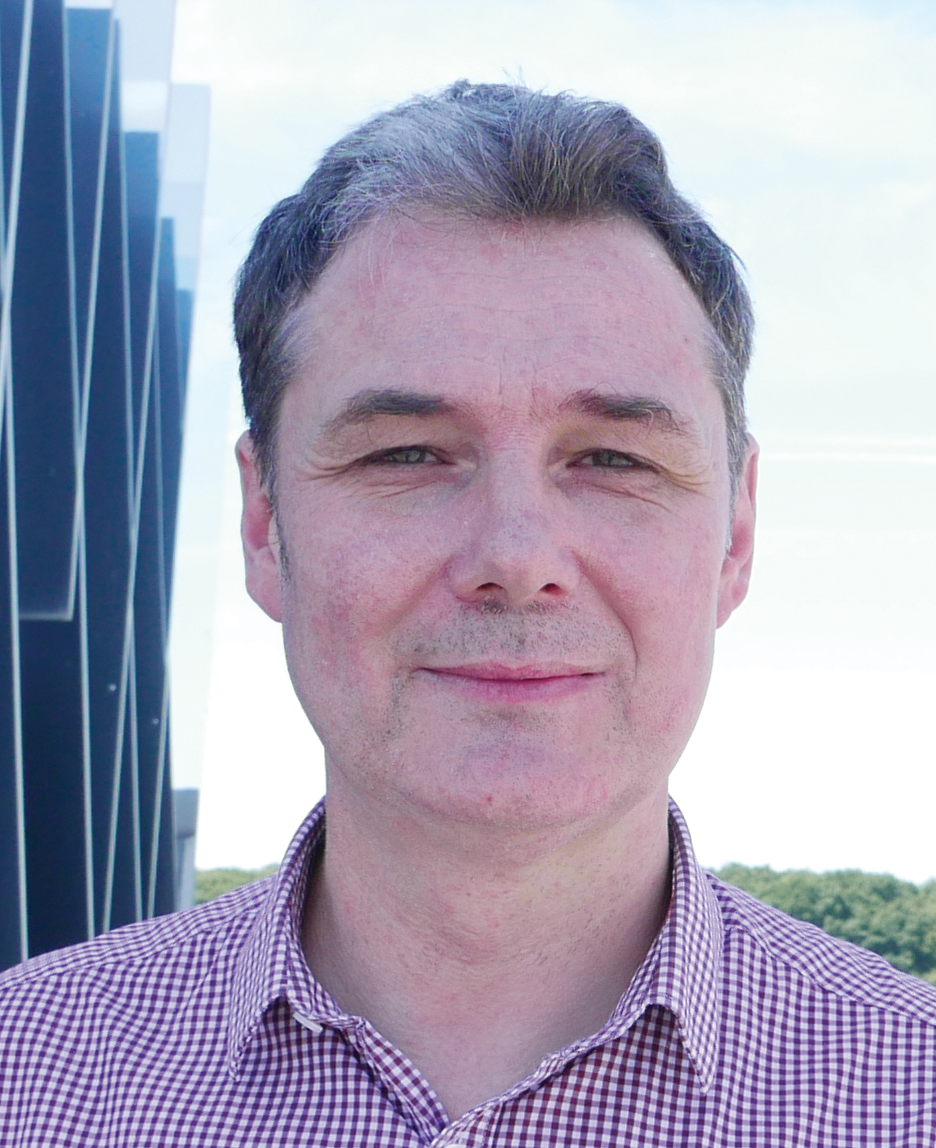}}]{St\'ephane Paquelet} received the B.Sc. degree
from the Ecole Polytechnique, Paris, France, in 1996, and the M.Sc. degree
from Telecom Paris, Paris, France, in 1998. He successively worked in
the fields of cryptology and signal processing for electronic warfare with
Thales and led UWB R\&D with Mitsubishi Electric from 2002 to 2007,
where he proposed two pioneering transceivers for short-range/high data
rates and large-range/low data rates, including telemetry. He developed multi-standard reconfigurable RF-IC at Renesas Design France until 2014. From 2014 to 2019, he led wireless communications activities for IRT b\raisebox{0.2mm}{\scalebox{0.7}{\textbf{$<>$}}}com, where he now leads the artificial intelligence laboratory.
\end{IEEEbiography}




\end{document}